%% file: main.tex
\documentclass[conference]{IEEEtran}
\IEEEoverridecommandlockouts
\usepackage{cite}
\usepackage{amsmath,amssymb,amsfonts}
\usepackage{algorithmic}
\usepackage{graphicx}
\usepackage{textcomp}
\usepackage{xcolor}
\def\BibTeX{{\rm B\kern-.05em{\sc i\kern-.025em b}\kern-.08em
    T\kern-.1667em\lower.7ex\hbox{E}\kern-.125emX}}
    
\usepackage[vlined,commentsnumbered,ruled,linesnumbered]{algorithm2e}
\let\oldnl\nl% Store \nl in \oldnl
\newcommand{\nonl}{\renewcommand{\nl}{\let\nl\oldnl}}% Remove line number for one line
\usepackage[skip=0pt]{caption}
\usepackage{subcaption}
\usepackage{xcolor}
\usepackage{amsmath}
\usepackage{enumitem}
\usepackage{url}

\input{macros}
\usepackage{paralist}

\usepackage{multirow}

\begin{document}

\title{Patterns Count-Based Labels for Datasets}

\author{\IEEEauthorblockN{Yuval Moskovitch}
\IEEEauthorblockA{
\textit{University of Michigan}\\
yuvalm@umich.edu}
\and
\IEEEauthorblockN{H. V. Jagadish}
\IEEEauthorblockA{
\textit{University of Michigan}\\
jag@umich.edu}
}

\maketitle

\begin{abstract}
Counts of attribute-value combinations are central to the profiling of a dataset, particularly in determining fitness for use and in eliminating bias and unfairness. While counts of individual attribute values may be stored in some dataset profiles, there are too many combinations of attributes for it to be practical to store counts for each combination.  In this paper, we develop the notion of storing a ``label" of limited size that can be used to obtain good estimates for these counts.
	A label, in this paper, contains information regarding the count of selected patterns--attributes values combinations--in the data. We define an estimation function, that uses this label to estimate the count of every pattern.
	We present the problem of finding the optimal label given a bound on its size and propose a heuristic algorithm for generating optimal labels. We experimentally show the accuracy of count estimates derived from the resulting labels  and the efficiency of our algorithm. 
% 	We present the problem of finding the optimal label given a bound on its size, show that this problem is NP-hard and propose a heuristic algorithm for generating optimal labels. We experimentally show the accuracy of count estimates derived from the resulting labels  and the efficiency of our algorithm. 
\end{abstract}

% \begin{IEEEkeywords}
% component, formatting, style, styling, insert
% \end{IEEEkeywords}

\input{intro}
\input{model}

\input{algo}

\input{exp}

\input{related}

\input{conc}

% \bibliographystyle{IEEEtran}
\bibliographystyle{plain}
\small{
\bibliography{bibtex}
}
% \clearpage
\appendix
\input{reduction}

\end{document}

%% file: macros.tex
\DeclareMathOperator*{\argmin}{arg\,min}

\newtheorem{theorem}{Theorem}[section]
\newtheorem{lemma}[theorem]{Lemma}
\newtheorem{proposition}[theorem]{Proposition}
\newtheorem{corollary}[theorem]{Corollary}

\newtheorem{example}[theorem]{Example}

\newtheorem{definition}[theorem]{Definition}
\newtheorem{proof}[theorem]{Proof}

\newcommand{\yuval}[1]{}

%% file: intro.tex
%!TEX root = ./main.tex

\section{Introduction}

Data-driven decision systems are increasingly used today. The data on which these systems depend, as in much of data science, are often ``found data", namely, data that was not collected as part of the development of the analytics pipeline, but was rather acquired independently, possibly assembled by others for different purposes. 
% %With the emerging variety of publicly available datasets, and their online accessibility, it is common practice to use found data.
% In particular, when the decision is made by a machine-learned model, the correctness and quality of the decision depend centrally on the data used in the model training phase.
When the decision is made by a machine-learned model, the correctness and quality of the decision depend centrally on the data used in the model training phase.

% algorithms and tools are commonly used in different advanced image and speech recognition. An essential part of such algorithms is the 

The use of improper, unrepresentative, or biased data may lead to unfair decisions, algorithmic discrimination (such as racism), and biased models \cite{abs-1909-01866}. 
  %A well-known example in this context is google's photos application released in 2015 that categorized African American people as ``gorillas" \cite{mulshine_2015}.  
Data-driven methods are increasingly being used in domains such as fraud and risk detection, where data-driven algorithmic decision making may affect human life. For instance, risk assessment tools, which predict the likelihood of a defendant to re-offend, are widely used in courtrooms across the US~\cite{machineBias}. ProPublica, an independent, non-profit newsroom that produces investigative journalism in the public interest, conducted a study on the risk assessment scores output by a software developed by Northpointe, Inc. They found that the software discriminated based on race: blacks were scored at greater risk of re-offending than the actual, while whites were scores at lower risk than actual.

Further analysis \cite{AsudehJJ19} showed issues with other groups as well. For example, the error rate for Hispanic women is very high because there aren't many Hispanic women in the data set. It is not only that there are fewer Hispanics than blacks and whites, and fewer women then men, but also fewer Hispanic women than one would expect if these attribute values were independently distributed. A judge sentencing a Hispanic woman presumably would like to be informed about this low count of Hispanic women in the data set and the consequent likelihood of greater error in the risk assessment.
  
%The use of improper, unrepresented or biased data, unrepresented data, data with inadequate coverage of subgroups or biased data may lead to discriminating or unfair decisions, algorithmic racism and to biased models \cite{abs-1909-01866}. A well-known example  in this context is google's Photos application released in 2015 that categorized African American people as ``gorillas" \cite{mulshine_2015}.  The damage cased by racist and biased algorithm tends to be more severe in vulnerable domains where data-derive algorithmic decision making may effects human life. For instance, risk assessments tools, that predict the likelihood of defendant to re-offend are widely used in courtrooms across the US~\cite{machineBias}. ProPublica, an independent, non-profit newsroom that produces investigative journalism in the public interest, has launch a study on a risk assessment results of software developed by Northpointe, Inc. Their results were harassing, showing a bias and race-based discrimination.   

When using ``found data", analysts typically perform data profiling, a process of extracting metadata or other informative summaries of the data \cite{AbedjanGN15}. Examples of information acquired in this process include statistics over the attributes' values, their type, common patterns, and attributes correlations and dependencies. Such information may assist in mitigating the misuse of data and reduce algorithmic bias and racism. While informative and useful, data profiling is hard to do well, is usually not automated, and requires significant effort.
%some profiling tasks require dedicated tools, may take hours (depending on the task) and sometime require results interpretation, which usually done by experts.

Even users of the data (or data analysis), and not just the analysts, may be interested in this sort of profiling information on the training data before they can trust the learned model.  To help both the data analyst and the data user, the notion of a ``nutrition label" has been suggested \cite{StoyanovichH19, DNL, datasheets, MithraLabel, modelCards,NLforRanking}.
The basic idea of a nutrition label is to capture, in a succinct label, data set properties of interest.  
Perhaps the single most important such property is a profile of the counts of various attribute value combinations.
For instance, an analyst may wish to ensure a (close) to real-world distribution in the attribute's values of the data, such as an equal number of males and females. Another concern may be the lack of adequate representation in the data for a particular group \cite{AsudehJJ19}, such as divorced African-American females, or contrarily, a high percentage of data that represents the same group (data skew) \cite{ChenJS18}. The count information may also reveal potential dependent or correlated attributes. As a simple example, if all tuples representing individuals under 20 years old are also single, this may point out a possible connection between age and marital status.

%\yuval{from Chris' feedback: In the bigger picture, I'm wondering whether it is the most useful approach to provide the count-based "label" proposed here as metadata with the dataset. Why would we not go a step further and actually preprocess and discover the "dangerous intersected attribute combinations"  (without counts) and include it as metadata warnings? Would that consume too much storage? I'm not sure about this thought, but I'm just pondering the possibility...}
%\yuval{Maybe we can say the following} 
Of course, interpretation of the count information depends on the intended use of the data set. Users performing different tasks may be interested in various parts of the data and their counts. Moreover, the thresholds set for skew or inadequate data may vary for different uses. Once the count information is available, it can be used to develop usecase-specific metadata warnings such as ``dangerous intersected attribute combinations" or ``inadequate representation of a protected group".

 In this paper, we propose to label datasets with information regarding the count of different patterns (attributes values combinations) in the data,  which can be useful to determine fitness for use.
Needless to say, there is a combinatorial number of such combinations possible.  So, storing individual counts for each is likely to be impossible.  
To this end, we focus on techniques to estimate these counts based on storing only a limited amount of information.
  
% Such information may be used  to determine fitness for use regarding the  satisfaction of fairness constraint, detecting underrepresented groups and potentially dependent or correlated attributes.

 \begin{example}\label{ex:intro}
 COMPAS is a risk assessment commercial tool made by Northpointe, Inc.
The COMPAS dataset was collected and published by ProPublica.
 as part of their investigation~\cite{compas}. 
 The full dataset contains  60,843 tuples with 29 attributes, including meaningful demographic groups such as gender, race, age, marital status, assessment reason, agency (e.g., pretrial, probation), language, legal stats, custody status, and supervision level. 
Four of these attributes are shown in a fragment of a simplified version of the dataset in Figure~\ref{fig:db}.  
    Partial counts information of the simplified version is given in Figure \ref{fig:label}.
  %, including only four attributes: gender, age group, race and marital status.
  This dataset description depicts the possible values of each attribute, and their count in the data, with the addition of counts for some attribute value combinations: gender and race in this example. Some immediate observations that can be made based on this information is that female and male are not equally represented in the data, and due to the low number of widows in the data, there is a high possibility that the number of Hispanic female widows is inadequate for the development of non-biased algorithm using this data.   
 \end{example}

\setlength{\belowcaptionskip}{-20pt}
 \begin{figure}
	\centering
	\includegraphics[scale=0.5]{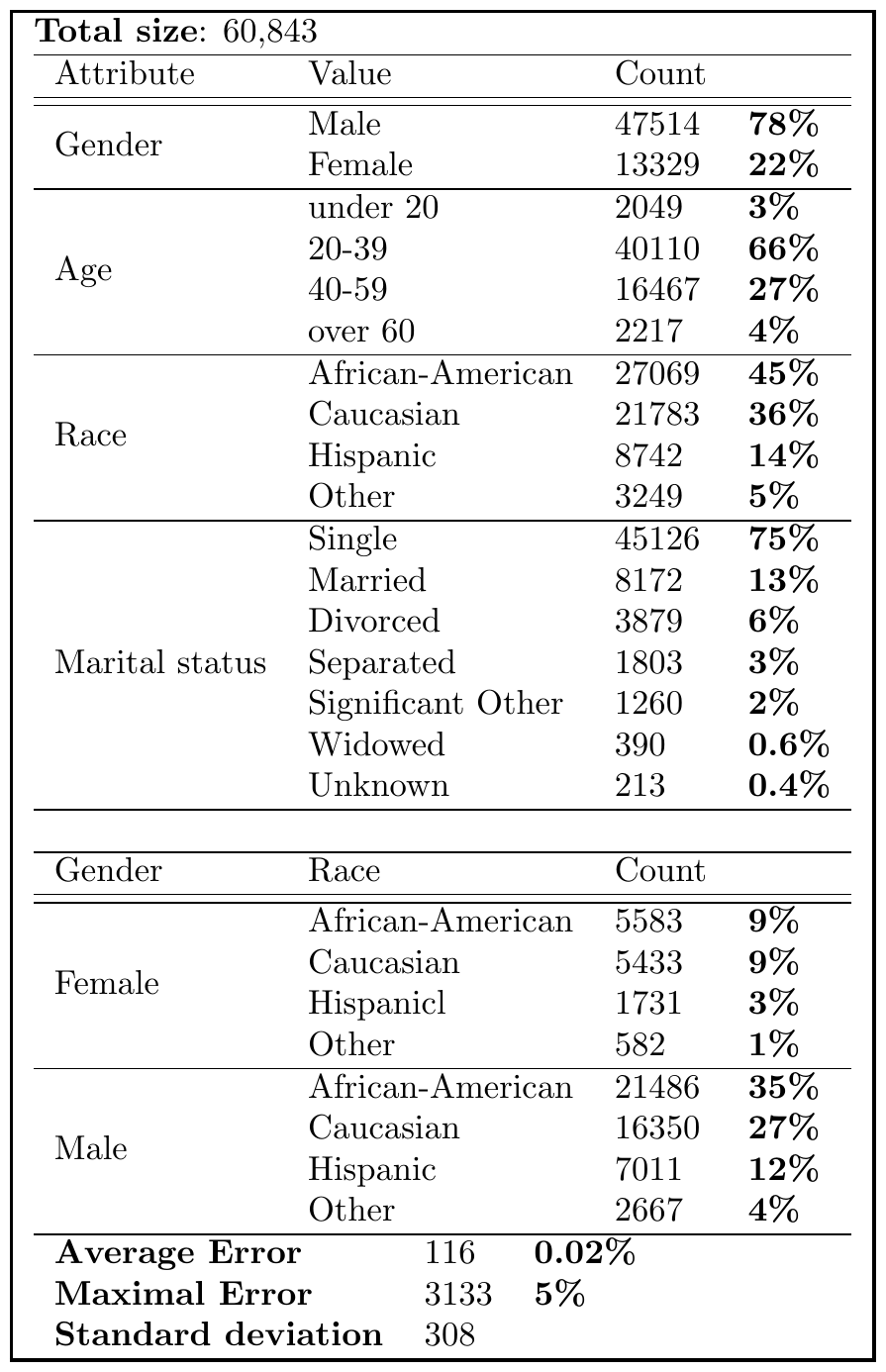}
	\caption{Labels computed for (a simplified version of the) COMPAS dataset}
	\label{fig:label}
\end{figure}

Given a data set, if we do not know anything about value distributions in it, a common assumption to make is that of independence between attributes.  One way we could control the size of stored information is to keep counts for only individual attribute values, and to estimate counts for attribute value combinations, assuming independence. However, this defeats the central purpose of profiling -- we only get information about individual attributes (the ``marginal distributions") but nothing about any correlations. 
%But this defeats the central purpose of data profiling.
In the study of discrimination, there is a considerable examination of {\em intersectionality}, the whole point of which is to understand how the social consequence of being a member of a protected class on multiple axes is not simply the ``sum" of each alone. For example, to understand the discrimination faced by black women it is not enough to understand independently the impact of race alone and gender alone. In other words, we have to ensure that our estimates for the count of any pattern in the database are at least approximately correct.

Histograms have long been used for similar purposes in relational databases,  however, they do not do very well in high dimensions. Other prevalent techniques for selectivity estimation includes sampling, and machine learning-based methods (see review in Section~\ref{sec:related}). The former suffers from insufficient performance in the presence of skews and high selectivity queries, and the latter requires training and result in very complex models. Inspired by the concept of nutrition labels for datasets, a key requirement in our problem context is that the metadata annotation can be immediately comprehensible to a potential user of the dataset.

% We review histograms, sampling, and other prevalent techniques for selectivity estimation in Section \ref{sec:related}. Inspired by the concept of  nutrition labels for datasets, a key requirement in our problem context is that the metadata annotation be immediately comprehensible to a potential user of the dataset.

% See Section~\ref{sec:related} for a review of histograms, sampling, and other prevalent techniques.
% A key requirement in our problem context is that the metadata annotation be immediately comprehensible to a potential user of the data set: our inspiration is the concept of a nutrition label for a data set, a concept promoted independently by multiple parties in \cite{}...  

Our problem, intuitively, is to  choose a small number of patterns (limited by a given space budget), among the exponential number, that can be used to estimate the count for any pattern with minimal error. 
We envisage this information being made available as meta-data with each data set. 
In deference to the idea of a nutrition label, we call our stored information a {\em ``label"}.
An important feature of our model that is missing in previously  proposed models for data labeling is the ability to generate the labels in a fully automated manner. 

% Our problem, intuitively, is to  choose a small number of patterns among the exponential number, that can be used to estimate the count for any pattern with minimal error, where the number of patterns is limited by a given space budget. 
% We envisage this information being made available as meta-data with each data set. 
% In deference to the idea of a nutrition label, we call our stored information a {\em ``label"}.
% An important feature of our model that is missing in previously  proposed models for data labeling is the ability to generate the labels in a fully automated manner. 

% We note that the nutritional labels were proposed to improve transparency, accountability and fairness in data science, facilitate in increasing  the reproducibility of machine learning models, avoid data misuse and unwanted biases and help users select appropriate datasets for their chosen tasks. Different models were proposed for datasets labels~\cite{DNL, datasheets, MithraLabel} and for models labels~\cite{modelCards, NLforRanking}, and the count-based data description that we proposed in this paper can serve as a data nutritional label model. An important feature of our model that is missing in previous proposed models  is the ability to fully automatically generate the labels. 
% In deference to the idea of a nutrition label, we call our stored information a {\em ``label"}.

We define our notion of data labels with respect to a subset of attributes $S$, as the count information of all possible values combination of attributes in $S$ appearing in the data. The size of the label is then determined by the space required for the count information. 
By making an independence assumption, individual attribute value counts can be used to estimate the joint distribution, but if we are additionally given selected intersection counts, how should we use these to estimate other intersection counts not provided? 
%If we know the marginal distributions (or individual attribute value counts), we can make an independence assumption and estimate the joint distribution (multi-attribute intersection counts); but if we are additionally given selected intersection counts, how should we use these to estimate other intersection counts not provided?  
We present a model for this estimation in Section \ref{sec:model}.  Given the estimation procedure, each label entails an error with respect to the real count of patterns in the data. 
The problem of finding an optimal label within a given bound on the label size is NP-hard.

% We prove that the problem of finding an optimal label within a given bound on the label size is NP-hard in \cite{full}.

% We note that the nutritional labels \cite{StoyanovichH19} were proposed to improve transparency, accountability and fairness in data science, facilitate in increasing  the reproducibility of machine learning models, avoid data misuse and unwanted biases and help users select appropriate datasets for their chosen tasks. Different models were proposed for datasets labels~\cite{DNL, datasheets, MithraLabel} and for models labels~\cite{modelCards, NLforRanking}, and the count-based data description that we proposed in this paper can serve as a data nutritional label model. An important feature of our model that is missing in previous proposed models  is the ability to fully automatically generate the labels. 

%We note that the use of nutritional labels \cite{StoyanovichH19} was proposed to facilitate in increasing  the reproducibility of machine learning models, assist in avoidance of dataset misuse and unwanted biases and help users select appropriate datasets for their chosen tasks. Different models were proposed for datasets labels \cite{DNL, datasheets, MithraLabel} and for models labels~\cite{modelCards, NLforRanking}. The count-based data description that we proposed in this paper can serve as a data nutritional label model. An important feature of our model that is missing in previous proposed models  is the ability to fully automatically generate the labels. 

A naive algorithm for the problem would traverse over all possible attributes subsets in increasing size order, compute the size of the corresponding label for each set, and choose the one that entails the minimal error within the space budget. We argue that in practice, the labels generated with a set of attributes $S$ is preferable over labels generated using any subset of $S$, and build upon this property an optimized heuristic for the problem of finding an optimal label (Section~\ref{sec:algo}).

% A naive algorithm for the problem would traverse over all possible attributes subsets in increasing size order, for each set compute the space required for the count information of all possible values combinations for attributes in the set, and choose the one that entails the minimal error within the space budget. We argue that in practice, the labels generated with a set of attributes $S$ is preferable over labels generated using any subset of $S$, and build upon this property an optimized heuristic for the problem of finding an optimal label (Section~\ref{sec:algo}).

We conduct an extensive experimental study (Section \ref{sec:exp}) to assess the quality of our proposed labels model and the labels generation algorithm's performance using real-world datasets. Our experimental results demonstrate the high accuracy of the labels generated, even with a very limited space budget, and indicate the usefulness of our proposed optimized heuristic compared to the naive algorithm. They further show the scalability of the algorithm with respect to the generated label size, the data size, and the number of attributes.

We survey related work in Section \ref{sec:related} and conclude in Section \ref{sec:conc}.  %Due to space constraints, all proofs are deferred

%% file: model.tex
\section{Labels and Pattern Count Estimation}
%\section{Bound size labels}
\label{sec:model}

\setlength{\belowcaptionskip}{-15pt}
\begin{table}
    \centering
    \scriptsize
    \begin{tabular}{ll}
        \hline
        $D$ & Dataset \\
         $\mathcal{A}$&  Attributes set in $D$\\ 
         $Dom(A_i)$ & Active domain of attribute $A_i$ \\
         $p$ & Pattern \\
         $Attr(p)$ & The set of attributes in the pattern $p$ \\
         $c_D(p)$ & The count of tuples in $D$ satisfying $p$ \\
         $S$&  A subset of attributes ($S\subseteq\mathcal{A}$)\\ 
         $P_S$&  The set of all possible patterns over $S$ s.t. $C_D(p)>0$\\ 
         $L_S(D)$&  A label of $D$ using $S$\\ 
         $VC$ & The value count of each value in $D$ \\
         $PC$ & The pattern count of each tuples in $P_S$\\
         $p|_{S_1}$ & The pattern resulting when restringing $p$ to $S_1$\\
         $Est(p, l)$ & The estimation of a pattern $p$ using the label $l$\\
         $Err(l,p)$ & The error of $l$ with respect to $p$\\
         $\mathcal{P}$ & A set of patterns\\
         $Err(l,\mathcal{P})$ & The maximal error of $l$ with respect to $p\in \mathcal{P}$\\
         \hline
    \end{tabular}
    \caption{Notation Table}
    \label{tab:notations}
\end{table}

In this section we present a novel model of label construction, based on counts. A summary of the notations used throughout the paper is shown in Table \ref{tab:notations}.
We assume the data is represented using a single relational database, and that the relation's attributes values are categorical. Where attribute values are drawn from a continuous domain, we render them categorical by bucketizing them into ranges: very commonly done in practice to present aggregate results.  In fact, we may even group categorical attributes into fewer buckets where the number of individual categories is very large.% for some attribute.
%\yuval{should we say something about bucketization for non categorical data?}

%We start by presenting a new model of nutrition label which is base on counts.
%We assume the data is represented using a single relational database, and that the relation's attributes values are categorical. We first define the notion of pattern which plays a key rule in our labeling model. 

\subsection{Patterns count information}
We first define the notion of pattern which is the foundation for our label model. 

\begin{definition}[Patterns]
	Let $D$ be a database with attributes $\mathcal{A}=\{A_1,\ldots, A_n\}$ and let $Dom(A_i)$ be the active domain of $A_i$ for $i \in [1..n]$. A \emph{pattern} $p$ is a set  $\{A_{i_1} = a_1, \ldots, A_{i_k} = a_k\}$ where $\{A_{i_1},\ldots, A_{i_k}\}\subseteq \mathcal{A}$ and $a_j\in Dom(A_{i_j})$ for each $A_{i_j}$ in~$p$. We use $Attr(p)$ to denote the set of attributes in $p$.
	
	%	A \emph{full pattern} is a pattern $P$ that contains all $A_i\in \mathcal{A}$. We use $fp(D)$ to denote the set of all possible full patterns in the database $D$.
\end{definition}

\begin{example}
	Consider the  fragment of the simplified version of the COMPAS database given in Figure \ref{fig:db}. %It contains four categorical attributes: gender, age group, race and marital status.
% 	\{age group= under 20, marital status = singe\} is a possible pattern over  $S=$\{age group, marital status\}.  
	$p=$\{age group= under 20, marital status = singe\} is a possible pattern and  $Attr(p)=$\{age group, marital status\}. 
\end{example}	

\begin{definition}
	We say that a tuple $t\in D$ \emph{satisfies} a pattern $p$ if $t.A_i = a_i$ for each $A_i\in Attr(p)$. % appearing in $p$.
	The count $c_D(p)$ of a pattern  $p$ is the number of tuples in $D$ that satisfy~$p$. 
	%	We use $cp(D)$ to denote the set of all possible categories patterns over $D$. % with count greater than~$0$.
\end{definition}

\begin{example}
	Consider again the database given in Figure \ref{fig:db}. The tuples 1, 3, 8, 10, 12, and 14 satisfy the pattern $p=$\{age group= under 20, marital status = single\} and thus the count of~$p$~is~$c_D(p)=6$.  
\end{example}	

\begin{figure}
	\centering
	\scriptsize
	\begin{tabular}{lcccc}
		%\cline{2-5}
		\hline
		&Gender & Age group & Race & Marital status\\ \hline % \cline{2-5}
		1&Female & under 20& African-American  & single\\
		2	&Male & 20-39& African-American & divorced\\
		3&Male & under 20& Hispanic  & single\\
		4&Male & 20-39& Caucasian & married\\
		5&Female & 20-39& African-American & divorced\\
		6	&Male & 20-39& Caucasian & divorced\\
		7&Female & 20-39& African-American & married\\
		8		&Male & under 20& African-American  & single\\
		9			&Female & 20-39& Caucasian & divorced\\
		10&Male & under 20& Caucasian  & single\\
		11&Male & 20-39& Hispanic & divorced\\
		12	&Female & under 20& Hispanic & single\\
		13&Female & 20-39& Hispanic & married\\
		14&Female & under 20& Caucasian & single\\
		15&Female & 20-39& Caucasian & married\\
		16	&Male & 20-39& Hispanic & married\\
		17&Male & 20-39& African-American & married\\
		18&Female & 20-39& Hispanic & divorced\\
		
		%			\cline{2-5}
		\hline
	\end{tabular}
	
	\caption{Sample data from a simplified version of the COMPAS dataset}
	\label{fig:db}
\end{figure}

Information regarding the count of patterns appearing in the data can be useful to determine fitness for use. It may be used to ensure a (close) to real world distribution in the attribute's values of the data to detect improper (underrepresented) or extremely high representation (data skew) of patterns, and potential dependent or correlated attributes.
While full count of each pattern provides detailed and accurate description of the data, it can be extremely large. In fact it can have the same size as the data.

%, and thus can not be used as a data label as we next demonstrate. 

%Information regarding the count of patterns appearing in the data can be useful to determine fitness for use.
%For instance, an analyst may wish to ensure a (close) to real world distribution in the attribute's values of the data, such as equal number of male and female. Another concern may involve the lack of adequate group's representation in the data, such as divorced African-American female, or contrarily, high percentage of the data that represents the same group (data skew). The patterns count information may also reveal potential dependent or correlated attributes. As a simple example, if all tuples representing individuals under 20 year old are also single, this may point out a possible connection between age and marital status. 

\begin{example}\label{ex:simpleDB}
	As a simple example, consider a database $D$ with $n$ binary attributes $A_1, \ldots, A_n$, where each value combination $(b_1,\ldots, b_n)$, for $b_i\in \{0,1\}$, appears exactly once. In this case the database, as well as the patterns count,  includes $2^{n}$ tuples. 
\end{example}

To this end, we propose an estimation function, which estimates a pattern count based on partial count information. Our basic intuition is that information regarding the count of individual attributes values is sufficient to provide a good estimate of any pattern count if there are no correlations within the attributes. 
%This assumption holds if there are no correlations within the attributes. %, and the values of each attribute are randomly distributed, 
\begin{example}\label{ex:ind_est}
	Continuing with Example \ref{ex:simpleDB}, given the counts $c_D(\{A_i = b_i\}) = \frac{2^n}{2}$, the count of the pattern $\{A_1 = 0, A_2 = 0, A_3 = 0\}$ may be estimated as 
%	\begin{multline*}
%	2^n\cdot \frac{c_D(\{A_1 = 0\}) }{c_D(\{A_1 = 0\})+ c_D(\{A_1 = 1\})} \cdot \\\frac{c_D(\{A_2 = 0\}) }{c_D(\{A_2 = 0\})+ c_D(\{A_2 = 1\})}= \\2^n\cdot\frac{1}{2} \cdot\frac{1}{2} = 2^{n-2}
%	\end{multline*}
	\begin{multline*}
2^n\cdot \prod_{i=1}^3\frac{c_D(\{A_i = 0\}) }{c_D(\{A_i = 0\})+ c_D(\{A_i = 1\})}= 2^n\cdot \Big(\frac{1}{2}\Big)^3= 2^{n-3}
\end{multline*}
	Intuitively, under the assumption that there are no correlations, 
	%and the attributes values are randomly distributed, 
	the count of the pattern $\{A_1 = 0, A_2 = 0, A_3=0\}$ is the relative portion of the data (total number of $2^n$ tuples), that have the value $0$  in the attribute $A_1$, $A_2$ and $A_3$, which is reflected in the sub-expressions $\frac{c_D(\{A_i = 0\}) }{c_D(\{A_i = 0\})+ c_D(\{A_i = 1\})}$ in the computation. 	
	 In general, the count of the pattern 
	$p = \{A_{i_1} = b_{i_1},\ldots, A_{i_k} = b_{i_k}\}$ can be computed as
	$$
	|D|\cdot\prod_{j = 1}^{k} \frac{c_D(\{A_{i_j} = b_{i_j}\}) }{c_D(\{A_{i_j} = 0\}) + c_D(\{A_{i_j} = 1\})}
	$$
\end{example}

When we introduce correlations,  
%or change the distributions of values, 
the counts of individual attributes are no longer sufficient to provide a good estimation, as we next demonstrate. 

\begin{example}
	As a simple example, consider a database $D$ with $n$ binary attributes as described in Example \ref{ex:simpleDB}, except that  the values in the attributes $A_1$ are replaced such that the value of $A_1$ is equal to the value of $A_2$ for every tuple. The real count of the pattern $\{A_1 = 0, A_2 = 0,  A_3 = 0\}$ is now $2^{n-2}$, where using only the individual count the pattern count estimation is $2^{n-3}$ with the same computation shown in Example \ref{ex:ind_est}.
\end{example}

We may remedy this problem by using additional count information. In the above example, the counts of the patterns $p = \{A_1 = b_1, A_2 = b_2\}$ for $b_i\in \{0,1\}$ is sufficient to provide an exact estimate for each pattern in the database.

\begin{example}
	Given the patterns count $c_D(\{A_1 = 0, A_2 = 0\})  = 2^{n-1}$  we can compute the count of $\{A_1 = 0, A_2 = 0, A_3=0\}$ as 
	\begin{multline*}
	2^{n-1}\cdot \frac{c_D(\{A_3 = 0\}) }{c_D(\{A_3= 0\})+ c_D(\{A_3 = 1\})}= 2^{n-1}\cdot\frac{1}{2} = 2^{n-2}
	\end{multline*}
	In general, the count of any pattern $p = \{A_{i_1} = b_{i_1},\ldots, A_{i_k} = b_{i_k}\}$  (that contains $\{A_1 = b_1, A_2 = b_2\}$ for $b_i\in \{0,1\}$) 
	can be computed as
	$$
	c_D(\{A_1 = b_1, A_2 = b_2\}) \cdot\prod_{j = 3}^{k} \frac{c_D(\{A_{i_j} = b_{i_j}\}) }{c_D(\{A_{i_j} = 0\}) + c_D(\{A_{i_j} = 1\})}
	$$

\end{example}

Real world datasets are typically complex, and  have correlations among attributes. One possible way to tackle this problem is to store more information about these (large) deviations from our initial independence assumption. The challenge is to spend wisely a limited space budget to capture exactly the deviations that induce greatest error in our estimates.

\subsection{Patterns count based labels}

We next define our notion of data label. 
A label is defined with respect to a subset $S$ of the database attributes, and it contains the pattern count ($PC$) for each possible pattern over $S$ and value count ($VC$) of each value appearing in $D$. Given a subset of attributes $S\subseteq \mathcal{A}$ we use $P_S$ to denote the set of all possible patterns over $S$ (i.e., $p$ with $Attr(p) = S$) such that $c_D(p)>0$. The maximal number of patterns in $P_S$ is $\prod_{A_i \in S} |Dom(A_i)|$. 

\begin{definition}[Label]
	Given a database $D$ with attributes $\mathcal{A} = \{A_1,\ldots, A_n\}$, and a subset of attributes $S\subseteq\mathcal{A}$ a \emph{label}  $L_S(D)$ of $D$ using $S$  contains the set $PC = \{(p_i,c_D(p_i))\}$ for each $p_i\in P_S$ and the set $VC = \{(\{A_i = a_j\}, c_D(\{A_i = a_j\}))\}$ for each $A_i \in \mathcal{A}$ and $a_j\in Dom(A_i)$.
%	We use $|L_S(D)|$ to denote the label size, which is measured  by the number of pairs it contains.
\end{definition}

\begin{example}\label{ex:label}
	Consider the database fragment given in Figure \ref{fig:db}, the label resulting from use of the attributes set $S$ = \{age group, marital status\} consists of the following:
	\vspace{-0.1cm}
	\begin{equation*}
	\begin{split}
		PC = \{&(\{\text{age group = under 20, marital status = single}\},6)\\
		&(\{\text{age group = 20-39, marital status = married}\},6),\\
		&(\{\text{age group = 20-39, marital status = divorced}\},6)\}\\
		VC = \{&(\{\text{gender = female}\},9), (\{\text{gender = male}\},9),\\
		&(\{\text{age group = under 20}\},6),\\
		&(\{\text{age group = 20-39}\},12),\\
		&(\{\text{race = African-American}\},6),\\
		&(\{\text{race = Hispanic}\},6),\\
		&(\{\text{race = Caucasian}\},6),\\
		&(\{\text{marital status = single}\},6),\\
		&(\{\text{marital status = divorced}\},6),\\
		&(\{\text{marital status = married}\},6)\}
	\end{split}
\end{equation*}
The label resulting from use of the attributes set $S'$ = \{gender, age group\} consists of the same $VC$ set and the following $PC$~set:
	\begin{equation*}
		\begin{split}
		PC = \{&(\{\text{gender = female, age group = under 20}\},3)\\
		&(\{\text{gender = male, age group = under 20}\},3),\\
		&(\{\text{gender = female, age group = 20-39}\},6),\\
		&(\{\text{gender = male, age group = 20-39}\},6)\}
		\end{split}
		\end{equation*}
\end{example}	

% Note that for a given database $D$,  the $VC$ set is similar in every label of $D$. We next explain how the data labels can be used to estimate the count of every pattern in the database. 

% As we show in the sequel, the $VC$ set information is an integral part of our estimation method.

Note that for a given database $D$,  the $VC$ set is determined and similar for every label of $D$. This set may be large, for instance, the COMPAS dataset includes at leas 10 meaningful demographic attributes as shown in Example~\ref{ex:intro} and the Credit Card dataset~\cite{credit_card_data} we used in our experiments has over 20 attributes, including demographic factors, credit data and history of payments (see Section~\ref{sec:exp} for more details). As we show in the sequel this information is an integral part of the estimation method we propose. However, note that with a simple user interface, the label's presentation may be manually refined and attributes can be filtered-out in order to adjust the information to the user's interest.

% As we show in the sequel this information is an integral part of the estimation method we propose and thus can not be removed from our label definition.  Moreover, as our goal is developing general purpose labels, we do not restrict the value counts presented to the used by default. We note that with a simple user interface, label presentation may be manually adjusted and filter-out by users.% in any system that allows for label presentation \yuval{cite the demo?}.

% \yuval{new. Not sure if this is the right place to put it}
% Note that for a given database $D$,  the $VC$ set is is determined and similar in every label of $D$.
% This set may be large (e.g., the full COMPAS data has 17 attributes), although only a few attributes (e.g., gender, race) may be related to a fairness application. As we show in the sequel this information is an integral part of the estimation method we propose and thus can not be removed from our label definition. Moreover, as our goal is developing general purpose labels, we do not restrict the value counts presented to the used by default. The label presentation may be manually adjusted and filter-out by users using simple user interface in any system that allows for label presentation \yuval{cite the demo?}.% (such as \cite{demo}).

Let $D$ be a database with attributes $\mathcal{A}$,  and $S_1$ and $S_2$ be two subsets of attributes such that $S_1\subseteq S_2\subseteq \mathcal{A}$ . Given a pattern $p\in P_{S_2}$, we use $p|_{S_1}$ to denote the pattern that results when $p$ is restricted to include only the attributes of $S_1$. Given a label of $D$ using $S_1$, we may estimate the count of each pattern in $P_{S_2}$ as follows.

\begin{definition}[Pattern Estimation]
	Let $D$ be a database with attributes $\mathcal{A}$ and $S_1\subseteq S_2\subseteq \mathcal{A}$ be two subsets of attributes.
	Given a label $l = L_{S_1}(D)$ the count estimate for a pattern $p\in P_{S_2}$ is
	$$
	Est(p,l) = c_D(p|_{S_1})\cdot\prod_{A_i\in S_2\setminus S_1}\frac{c_D(\{A_i = p.A_i\})}{\sum_{a_j\in Dom(A_i)} c_D(\{A_i = a_j\})}
	$$
	
\end{definition}

\begin{example}\label{ex:est}
	Consider again the database given in Figure~\ref{fig:db}, and the label $l = L_S(D)$ generated using $S = $\{age group, marital status\} shown in Example \ref{ex:label}. The estimate of the pattern $p=$\{gender = female, age group = 20-39, marital status = married\} using $l$ is 
	\begin{multline*}
	Est(p,l) =  \\c_D(\text{age group = 20-39, marital status = married})\cdot\\
	\frac{c_D(\{\text{gender = female}\})}{\sum_{a_j\in Dom(\text{gender})} c_D(\{\text{gender} = a_j\})} = 6\cdot\frac{9}{18} =3
	\end{multline*}	
%	\begin{equation*}
%	\begin{split}
%	Est(p,l) = & c_D(\text{age group = 20-39, marital status = married})\cdot\\
%	&\frac{c_D(\{\text{gender = female}\})}{\sum_{a_j\in Dom(\text{gender})} c_D(\{\text{gender} = a_j\})} = 6\cdot\frac{9}{18} =3
%	\end{split}
%	\end{equation*}
	Using the label $l' = L_{S'}(D)$ generated from $S'=$ \{gender, age group\}, with a similar computation we obtain 
	\begin{multline*}
	Est(p,l') =  c_D(\text{gender = female, age group = 20-39})\cdot\\
	\frac{c_D(\{\text{marital status = married}\})}{\sum_{a_j\in Dom(\text{marital status})} c_D(\{\text{marital status} = a_j\})} = \\
	6\cdot\frac{6}{18} =2
	\end{multline*}
%	\begin{equation*}
%	\begin{split}
%	Est(p,l') = & c_D(\text{gender = female, age group = 20-39})\cdot\\
%	&\frac{c_D(\{\text{marital status = married}\})}{\sum_{a_j\in Dom(\text{marital status})} c_D(\{\text{marital status} = a_j\})} = \\
%	&6\cdot\frac{6}{18} =2
%	\end{split}
%	\end{equation*}
\end{example}	

%\begin{example}\label{ex:est}
%	Consider again the database given in Figure \ref{fig:db}, and the label $l = L_S(D)$ generated using $S = $\{age group, marital status\} shown in Example \ref{ex:label}. The estimation of the pattern $p=$\{sex = female, age group = 20-39,race = Hispanic, marital status = married\} using $l$ is 
%		\begin{equation*}
%	\begin{split}
%	Est(p,l) = & c(\text{age group = 20-39, marital status = married})\cdot\\
%	&\frac{c(\{\text{sex = female}\})}{\sum_{a_j\in Dom(\text{sex})} c(\{\text{sex} = a_j\})}\cdot\\
%	&\frac{c(\{\text{race = Hispanic}\})}{\sum_{a_j\in Dom(\text{race})} c(\{\text{race} = a_j\})} = 6\cdot\frac{9}{18}\cdot\frac{6}{18} =1
%	\end{split}
%	\end{equation*}
%	Using the label $l' = L_{S'}(D)$ generated from $S'=$ \{sex, age group\}, with a similar computation we obtain 
%	\begin{equation*}
%	\begin{split}
%	Est(p,l') = & c(\text{sex = female, age group = 20-39})\cdot\\
%	&\frac{c(\{\text{race = Hispanic}\})}{\sum_{a_j\in Dom(\text{race})} c(\{\text{race} = a_j\})}\cdot\\
%	&\frac{c(\{\text{marital status = married}\})}{\sum_{a_j\in Dom(\text{marital status})} c(\{\text{marital status} = a_j\})} = \\
%	&6\cdot\frac{6}{18}\cdot\frac{6}{18} =\frac{2}{3}
%	\end{split}
%	\end{equation*}
%	
%\end{example}	

We can then define the error of a label with respect to a pattern and a set of patterns.

\begin{definition}[Estimation Error]
	The error of a label $l = L_S(D)$ with respect to a pattern $p$ is 
	$$
	Err(l, p) = |c_D(p) - Est(p,l)|
	$$	
\end{definition}

\begin{example}
	Reconsider the estimates $Est(p,l)$ and $Est(p,l')$ of the pattern $p=$\{gender = female, age group = 20-39, marital status = married\}  shown in Example \ref{ex:est}. The count of the pattern $p$ in the database is $3$, thus the error of $l$ with respect to $p$ is $0$ and the error of $l'$ is $1$.
\end{example}	

%\begin{example}
%	Reconsider the estimations $Est(p,l)$ and $Est(p,l')$ of the pattern $p=$\{gender = female, age group = 20-39,marital status = married\} using the labels $l$ and $l'$, generated with $S = $\{age group, marital status\} and  $S'=$ \{gender, age group\} respectively, as shown in Example \ref{ex:est}. The count of the pattern $p$ in the database is $3$, thus the error of $l$ with respect to $p$ is $0$ and the error of $l'$ is $1$.
%\end{example}	

Abusing notation, we use $Err(l, \mathcal{P}) $, for a set of patterns  $\mathcal{P}$, to denote the maximum  error in the estimate for any individual pattern in $\mathcal{P}$.
% We choose to focus on the maximum error (rather than mean for instance), as this definition of error is stiffer and gives us a sense of the error ``bound" over a large number of patterns in the database.

\paragraph*{Error metric}\yuval{new}
There are multiple plausible error measures which can be classified into two groups: relative and absolute error measures. An example of relative error measure, commonly used  in the field of selection estimation (see, e.g., \cite{MullerMK18, DuttWNKNC19, YangLKWDCAHKS19}) is the proportion between the selectivity estimation and the true selectivity, called $q$-$error$~\cite{MoerkotteNS09}.
$$q\text{-}error(p) = \max\Big(\frac{c_D(p)}{est(p)}, \frac{est(p)}{c_D(p)}\Big)$$
The $q$-$error$ metric is relative, symmetric, and is usually preferred since it ``fairly" penalize low selectivity estimations.

% While the labels we defined in this section is used for pattern cordiality estimation, similarly to selectivity estimation, the ultimate goals of the estimation are slightly different. 
Selectivity estimation techniques are geared towards query optimizations, and relates to query plan quality \cite{MoerkotteNS09}, while our labels are designed to assist end users determine fitness for use. This difference plays a rule when choosing the error measure. We choose to focus on the absolute maximum error (rather than mean for instance), as this definition of error is stiffer and gives us a sense of the error ``bound" over a large number of patterns in the database. Our problem definition, its hardness and proposed solution holds also when using $q$-$error$, and we report the resulting $q$-$error$ of the generated labels in the out experiments (see Section \ref{subsec:acc}).

\subsection{Problem definition}

We are now ready to define the optimal label problem.

\begin{definition}[Optimal Label Problem]
	Given a database $D$, with attributes $\mathcal{A}$, a bound $B_s$ over the label size, and a set of patterns $\mathcal{P}$, the optimal label is 
	$$
	\argmin_{ S\subseteq\mathcal{A}} Err( L_S(D), \mathcal{P})\text{ such that } | P_{S}|\leq B_s
	$$
\end{definition}

\yuval{new} Intuitively, the set of patterns $\mathcal{P}$ may be defined as $P_{\mathcal{A}}$ (i.e., the set of all possible patterns that include all the attributes and every value for each attribute that appears in the data). In this case $|\mathcal{P}|=|D|$ and an optimal label would be one that minimizes the error with respect to the count of tuples in the data. Our problem definition is more flexible, and allows the user to define a different pattern set,  e.g., patterns that include only sensitive attributes.

To formally characterize the complexity of the optimization problem, we further need to define a corresponding
decision problem. We define it as the problem of determining the existence of a label with size limited by the given bound and error which does not exceed a given error bound.

\begin{definition}[Decision Problem]
	Given a database $D$, with attributes $\mathcal{A}$, a bound $B_s$ over the label size, a set of patterns $\mathcal{P}$, and an error bound $B_e$, determine if there is a label $L_S(D)$ with $|P_{S}| \leq B_s$ and $Err(L_S(D), \mathcal{P}) \leq B_e$

\end{definition}

We can show that (see proof in the appendix).
\begin{theorem}\label{prop:np-hard}
	The decision problem is NP-hard.
\end{theorem}

More complex approaches could consider overlapping combinations of patterns, derive best estimates from multiple labels, use partial patterns, and so on.  Such complex approaches are left to future work.

%% file: algo.tex
\section{Optimal Label Computation}
\label{sec:algo}

Given a database $D$ with attributes $\mathcal{A} = \{A_1, \ldots, A_n\}$ and a bound $B_s$, a naive algorithm for the optimal label computation would operate as follows: iterate over possible attributes sets, starting with set of size $2$. At each iteration, compute the set of all possible labels with a fixed size, namely, at the $i$'th iteration the algorithm generate the labels $\{L_{S_1}(D),\ldots, L_{S_k}(D)\}$, where each $S_j$ for $j\in[1..k]$ is a subset of attributes of size $i+1$. For each label generated, compute its size and error, and record the optimal label computed with size below the given bound. The algorithm terminates if the size of all the labels generated in the same iteration exceeds the bound (or when all possible subsets were generated). Intuitively, if every attribute subset of size $i$ leads to a label with size greater than the given bound, then, every label generated using any attributes subset of size $>i$ would also exceed the bound.
%Iterate from $i=2$ and up to $i=n$ (the number of attributes), and for each $i$ compute the set of all possible labels $\{L_{S_1}(D),\ldots, L_{S_k}(D)\}$, where each $S_j$ is a set of $i$ attributes from $\mathcal{A}$. For each label, the algorithm computes its size and error, and records the optimal label computed with size below the given bound. When reaching $i$ such that the size of all the labels with $i$ attributes exceeds $B_s$, the algorithm terminates and return the optimal label recorded.
%\jag{Why are we iterating and then stopping?  We can have labels of different numbers of attributes in the label set.  Please check above description}
%\yuval{I rephrased the description. Does it make sense now?}
%We next present a 
The naive algorithm is unacceptably expensive. Therefore we developed a much faster 
heuristic solution for the optimal label problem.

\subsection{Label estimation characterization}\label{sec:label_char}

We start by characterizing the count estimation for a given pattern using a given label. Let $D$ be a database with attributes $\mathcal{A}$, $S\subseteq\mathcal{A}$ an attributes set and $l=L_S(D)$ a label of $D$ using~$S$. 

\begin{definition}
	 Given a pattern $p$, we say that the  estimate of $p$ using $l$ is 
	 \begin{compactitem}
	 \item
	 an \emph{exact estimation} if $Est(l,p) = c_D(p)$, 
	 \item
	 an \emph{over estimation} if $Est(l,p) > c_D(p)$, and 
	 \item 
	 an \emph{under estimation} if $Est(l,p) < c_D(p)$.
	 \end{compactitem}
\end{definition}

Clearly, for every pattern $p$ if $Attr(p)\subseteq S$ then the estimate of $p$ using $l$ is an exact estimation. Moreover, we can show the following:

%\begin{proposition}\label{prop:subset_acc}
%	For every pattern $p$
%	such that $Attr(p)\not\subseteq S_2$ let $p' = p|_{Attr(p)\cap S_2}$ be the pattern resulting when restricting $p$ to include only the attributes appearing in $S_2$. If the estimation of $p'$ using $l_1$ is an over (under) estimation, and the estimation of $p$ using $l_2$ is an over (under) estimation then $Err(l_2, p)\leq Err(l_1,p)$.
%\end{proposition}

\begin{proposition}\label{prop:subset_acc}
	Given two attribute sets  $S_1\subseteq S_2\subseteq\mathcal{A}$ and $l_i=L_{S_i}(D)$ the labels of $D$ using $S_i$ for $i=1,2$ respectively, for every pattern $p$
	such that $Attr(p)\not\subseteq S_2$ let $p' = p|_{Attr(p)\cap S_2}$ be the pattern resulting when restricting $p$ to include only the attributes appearing in $S_2$.
	If the estimate of $p'$ using $l_1$ is an over (under) estimation, and the estimate of $p$ using $l_2$ is an over (resp., under) estimation then $Err(l_2, p)\leq Err(l_1,p)$.
\end{proposition}

\begin{example}
	Suppose we are interested in estimating the number of married Hispanic females under the age of 20 in the data. Proposition \ref{prop:subset_acc} states that if the estimation of a label $l_1$ consisting of the count for gender and age combinations leads to an over (or resp. under) estimation of the pattern $p'=$\{gender = female, age = under 20, marital status = married\}, and the estimation  using a label $l_2$ generated with the count of the gender, age and marital status leads to an over (or under) estimation of $p=$\{gender = female, age = under 20, race = Hispanic, marital status = married\}, then $Err(l_2, p)\leq Err(l_1,p)$.
\end{example}

Intuitively, for two attributes sets $S_1$ and $S_2$, if $S_1\subseteq S_2$ the label generated using $S_2$ has more details than the one generated using $S_1$. In fact, based on Proposition \ref{prop:subset_acc}, it is reasonable to assume that  the pattern's count estimation using $L_{S_2}(D)$ is more precise than the one using  $L_{S_1}(D)$. We show that this assumption indeed holds in practice in our experiment (see Section \ref{sec:sub_lable_acc}).

% Our proposed solution is based on the above observation. 
% We start by defining a lattice over the possible labels, and then show how it can be used to  compute the optimal label.

Our proposed solution is based on the above observation. Our algorithm is inspired by the Apriori algorithm~\cite{apriori} and the Set-Enumeration Tree for enumerating sets in a best-first fashion~\cite{Rymon92}.
We start by defining a lattice over the possible labels, and then show how it can be used to  compute the optimal label.

% Our proposed solution is based on the above observation. We use the techniques from the Apriori algorithm~\cite{apriori} and the Set-Enumeration Tree for enumerating sets in a best-first fashion~\cite{Rymon92}.
% We start by defining a lattice over the possible labels, and then show how it can be used to  compute the optimal label.
% \yuval{Is this enough? should we say more?}

%Our proposed solution is based on the assumption that for a database $D$ and two attributes sets $S_1$ and $S_2$, if $S_1\subseteq S_2$ than the error of the label obtain from $S_1$ is at most the error of the label generated using $S_2$. More formally, for every pattern set $\mathcal{P}$, $$Err(L_{S_2}(D), \mathcal{P})\leq Err(L_{S_1}(D), \mathcal{P})$$ 
%Intuitively, the label generated using $S_2$ has more details that the one generated using $S_1$, and thus the pattern's count estimation is more precise \yuval{show this is true in the experiments section?}.
%\yuval{try to prove a variant of this assumption}

\subsection{Labels lattice}

We define a labels lattice as follows.

\begin{definition}[Labels lattice]
	Given a database $D$ with attributes $\mathcal{A}$, let  $\mathcal{A}^*$ be the set of all possible subset of $\mathcal{A}$. The \emph{label lattice} of $D$ is a graph $G = (V, E)$, where $V= \mathcal{A}^*$ and 
	$E = \{\{S_1, S_2\}\mid S_1\subset S_2 \text{ and } \exists A_i \in \mathcal{A}\text{ s.t.  }S_1\cup\{A_i\} = S_2\}$.
	
	$S_1$ is a parent (child) of $S_2$ if there is an edge $\{S_1, S_2\}$ and $S_1\subset S_2$ ($S_2\subset S_1$).
\end{definition}

Intuitively, $S_1$ is a parent of $S_2$ if $S_2$ can be obtained from $S_1$ by adding a single attribute $A\in \mathcal{A}\setminus S_1$.
Figure \ref{fig:lattice} depicts the label lattice of the database given in Figure \ref{fig:db} ($g$, $a$, $r$ and $m$ are use as abbreviations for gender, age group, race and marital status).

\begin{figure}
	\centering
	\includegraphics[width = 0.7\linewidth]{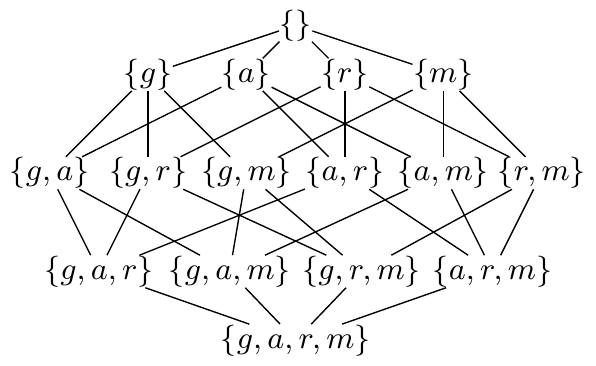}
	\caption{A label lattice}
	\label{fig:lattice}
\end{figure}

We note that, due to the nature and purpose of the labels (i.e., conciseness that allow for user friendly visualization), the typical bound over the label size is small. Thus, a natural way to scan the lattice is from the top down. Traversing the lattice does not require explicit representation of the graph, as children nodes can be generated on demand from their respective parents. Moreover we can generate each node in the label lattice exactly once in a top down scan as we next show. To this end we define the operator $gen(S)$ for a subset of attributes $S$ as follows.

\begin{definition}
	Let $D$ be a database with attributes $\mathcal{A}=\{A_1,\ldots, A_n\}$. We assume attributes are ordered, and for a given subset of attributes $S\subset\mathcal{A}$ we use $idx(S)$ to denote the index of the attribute with maximal attribute index in $S$, namely $idx(S) = \max_i(\{A_i\mid A_i \in S\})$, we define $$gen(S) = \{S'\mid S'=S\cup \{A_j\}~ \forall j~\text{s.t. } idx(S) < j \leq n\}$$
\end{definition}
For a given attributes set $S$, the set $gen(S)\subseteq children(S)$ where $children(S)$ is the set of all children of $S$ in the label lattice of $D$.
\begin{example}
	For the database $D$ given in Figure \ref{fig:db} and the attributes subset  $S=$\{gender, race\}, $gen(S)$ is \{gender, race, marital status\}. Note that \{gender, age group, race\} is a child of $S$ in the labels lattice, but is not included in $gen(S)$.
\end{example}

\subsection{Top down algorithm}
\yuval{made some slight changes}
%\paragraph*{Top-down algorithm} 
Algorithm \ref{algo:top_down} finds the optimal label using a top down traversal of the label lattice. The algorithm gets as input a database $D$, a set of patterns, and a bound $B_s$.
It uses a queue $Q$ to generate a candidate list of attributes subset, $cands$, such that the size of the label generated using each candidate in the list does not exceed the bound $B_s$. 

The algorithm first initializes the queue $Q$ with the set of attribute's singletons using $gen(\{\})$ (line \ref{line:initQ}), and the candidates set $cands$ to an empty set (line \ref{line:initCands}). Then while the queue $Q$ is not empty (lines \ref{line:whileStrart}~--~\ref{line:whileEnd}), the algorithm examines the first element in the queue $curr$ (line \ref{line:pop}). It traverses over the elements in $gen(curr)$ (lines \ref{line:forStart}~--~\ref{line:whileEnd}), and for each element $c$, checks if the size of the label generated by $c$ is not greater than $B_{s}$ (line \ref{line:if}). If so, the algorithm adds $c$ to the queue (line~\ref{line:addToQueue}) and update the candidates list, by removing the parents of $c$ that are currently in $cands$ (line \ref{line:remove}), and adding $c$ to the $cands$ list (line \ref{line:addCand}). 
Finally, the label that entails the minimal loss out of the set of all labels generated using  the attributes sets in the $cands$ list is returned (line \ref{line:return}).

\SetInd{1.3ex}{1.3ex}
\begin{algorithm}[t]
	\DontPrintSemicolon
	\SetKwInOut{Input}{input}\SetKwInOut{Output}{output}
	\LinesNumbered
	\Input{A database $D$, a set of patterns $\mathcal{P}$ and a bound $B_s$. }
	\Output{Optimal label.} \BlankLine
	\SetKwFunction{labelSize}{\texttt{labelSize}}
	\SetKwFunction{removeParents}{\texttt{removeParents}}
	\SetKwFunction{varCounts}{\texttt{varCounts}}
% 	$q = [\{\}]$ \label{line:initQ}\\
    % $B_{PC} \gets B_s - |\varCounts(D)|$\label{line:getBound}\\	
    $Q = [gen(\{\})]$ \label{line:initQ}\\	
	$cands = \emptyset$ \label{line:initCands}\\
	\While{$Q$ is not empty} 
	{\label{line:whileStrart}$curr\gets Q.dequeue()$ \label{line:pop}\\
		\For{$c \in gen(curr)$ \label{line:forStart}}
		{\If{\labelSize{$c$, $D$} $\leq B_{s}$ \label{line:if}}
			{$Q.enqueue(c)$ \label{line:addToQueue}	\\
				\removeParents{$cands, c$} \label{line:remove}\\
				$cands \gets cands\cup \{c\}$ \label{line:addCand}
				\label{line:whileEnd}
				}
			}
		}
	\Return $L_S(D)$ for $\argmin_{ S\in cands} Err( L_S(D), \mathcal{P})$\label{line:return}
	\caption{Top down search}\label{algo:top_down}
\end{algorithm} 

\yuval{new}
\begin{example}
Given the database $D$ shown in Figure \ref{fig:db}, the pattern's set $\mathcal{P}$ that contains the set of all tuples in $D$, and the bound $B_s = 5$, the algorithm first initializes $Q$ to be $[\{g\},\{a\}, \{r\}, \{m\}]$, and $cands$ to be an empty set. In the first iteration, $\{g\}$ is extracted from $Q$ and it's children, $\{\{g,a\},\{g,r\}, \{g,m\}\}$, are generated using $gen(\{g\})$. Out of this set, $\{g,a\}$ is the only subset that results in a label of size below $5$, and therefore is added to $Q$ and to $cands$. In the next iteration, $\{a\}$ is extracted from $Q$, and the algorithm examines the elements in $gen(\{a\} = \{\{a,r\},\{a,m\}\}$.  The label generated with $\{a,r\}$ is of size $3$ and the label generated with $\{a,m\}$ is of size $6$, thus only $\{a,r\}$ is added to $Q$ and $cands$. No other subset in the following iterations generates a label of adequate size, and the while loop terminates after all the elements in $Q$ are extracted. Finally, $cands$ contains $\{g,a\}$ and $\{a,m\}$, and the algorithm returns the label generated using $\{a,m\}$ since it is the optimal in this case.
\end{example}

% \yuval{new}
% \begin{example}
% Given the database $D$ shown in Figure \ref{fig:db}, the pattern's set $\mathcal{P}$ that contains the set of all tuples in $D$, and the bound $B_s = 15$, the algorithm first computes the bound over the pattern count set. The $VC$ set is shown in Example \ref{ex:label}, and $|VC| = 9$, thus, $B_{PC}=14-9=5$. Next, the algorithm initialize $Q$ to be $[\{g\},\{a\}, \{r\}, \{m\}]$, and $cands$ to be an empty set. In the first iteration, $\{g\}$ is extracted from $Q$ and it's children, $\{\{g,a\},\{g,r\}, \{g,m\}\}$, are generated using $gen(\{g\})$. Out of this set, $\{g,a\}$ is the only subset that results in a label of size below $5$, and therefore is added to $Q$ and to $cands$. In the next iteration, $\{a\}$ is extracted from $Q$, and the algorithm examines the elements in $gen(\{a\} = \{\{a,r\},\{a,m\}\}$.  The label generated with $\{a,r\}$ is of size $3$ and the label generated with $\{a,m\}$ is of size $6$, thus only $\{a,r\}$ is added to $Q$ and $cands$. No other subset in the following iterations generates a label of adequate size, and the while loop terminates after all the elements in $Q$ are extracted. Finally, $cands$ contains $\{g,a\}$ and $\{a,m\}$, and the algorithm returns the label generated using $\{a,m\}$ since it is the optimal in this case.
% \end{example}

By traversing the lattice in a top down fashion using the $gen$ operator the algorithm generates each node in the lattice at most once. Furthermore, the nodes generated are only attribute sets that lead to labels with size below the given bound, and (in the worst case) their children.

\begin{proposition}
	Given a database $D$, a set of patterns $\mathcal{P}$ and a bound $B_s$, Algorithm \ref{algo:top_down} generates each node in the label lattice at most once.
\end{proposition}

% \begin{proposition}
% 	Given a database $D$, a set of patterns $\mathcal{P}$ and a bound $B_s$, Algorithm \ref{algo:top_down} returns a label while generating each node in the label lattice at most once.
% \end{proposition}

Algorithm \ref{algo:top_down} avoids generating and exploring a large portion of the labels lattice, and in particular most of the labels that exceed the bound limit (which in practice are the majority, as shown by our experiments in Section \ref{subsec:optimization}).

%% file: exp.tex
\section{Experimental Evaluation}
\label{sec:exp}

We conducted experiments on real data to assess the quality of our proposed labels in estimating the data pattern's count. The key concerns are the size of label and the error in estimation.  We evaluated this trade off and considered the impact of data set parameters. We compared our label's accuracy to the performance of a real DBMS estimator, and the conventional approach of sample based estimation using different error measures.
A second issue we studied is the performance of the label generation algorithm.  We examined scalability in terms of label generation time as a function of (i) label's size bound, (ii) data size and (iii) number of data attributes.
We also quantified the usefulness of the heuristic approach compared to the naive algorithm.
Finally, we validated the assumption from Section \ref{sec:label_char} that more detailed labels lead to lower error.
In this section, we report on all these experiments in turn.  We begin with the set up we used.

% We conducted experiments on real and synthetic data to assess the quality of our proposed labels in estimating the data pattern's count. The key concerns are the size of label and the error in estimation.  We evaluated this trade off and considered the impact of data set parameters, such as number of attributes. 
% A second issue we studied is the performance of the label generation algorithm.  We examined the scalability of labels generating time as a function of (i) label's size bound, (ii) data size and (iii) number of data attributes (3)
% quantify the usefulness of the heuristic approach compared to the naive algorithm.
% Finally, we validated the assumption from Section \ref{sec:label_char} that more detailed labels lead to lower error.
% In this section, we report on all these experiments in turn.  We begin with the set up we used.

\subsection{Experimental setup}

%\paragraph*{Dataset} 
We used three real datasets with different numbers of tuples and attributes as follows.
\begin{compactdesc}
	\item[BlueNile] Blue Nile is an online jewelry retailer. We used the dataset collected and used in \cite{AsudehJJ19} of diamonds catalog, containing 116,300 diamonds. The dataset has 7 categorical
	attributes for the diamonds: shape, cut, color,
	clarity, polish, symmetry, and florescence.
	\item[COMPAS]  The COMPAS dataset was collected and published by ProPublica\cite{compas}. It contains 60,843 records that includes demographics, recidivism scores, and criminal offense information. The total number of attributes in the original database was 29. We removed id attributes (person id, assessment id, case id), names (first, last and middle), dates and attributes with less than 2 values or over 100 values. We added the attribute age, with four age ranges, based on the date of birth attribute. The resulting dataset contains 17 attributes.
%	\item[Compas]  The COMPAS dataset was collected and published by ProPublica.%, an independent, non-profit newsroom that produces investigative journalism in the public interest.	
%	COMPAS (which stands for Correctional Offender Management Profiling for Alternative Sanctions) is risk assessment commercial tool made by Northpointe, Inc.  ProPublica published
%	the dataset as part of their investigation into	racial bias in criminal risk assessment. It contains demographics, recidivism scores, and criminal offense information for 60,843 individuals. The total number of attributes in the original database was 29. We removed id attributes (person id, assessment id, case id), names (first, last and middle), dates and attributes with less than 2 values or over 100 values. We added the attribute age, with four age ranges, based on the date of birth attribute. The resulting dataset contains 17 attributes.
%	\item[Census-income  (KDD) Data Set] This data set contains weighted census data extracted from the 1994 and 1995 Current Population Surveys conducted by the U.S. Census Bureau. The full data contains 41 demographic and employment related variables and 299,285 tuples. We used \yuval{fill and justify}
\item[Default of Credit Card Clients Dataset\cite{credit_card_data}] This dataset contains information on default payments, demographic factors, credit data, history of payment, and bill statements of credit card clients in Taiwan from April 2005 to September 2005. It has 24 attributes and 30,000 tuples. We bucketize each numerical attribute into 5~bins. 
\end{compactdesc}
%In addition we used a synthetic dataset with up to 10 attributes, each with with up to 5 different values and 100,000 tuples.

In all the experiments we set $\mathcal{P}$, the patterns set, to be $P_\mathcal{A}$ where $\mathcal{A}$ in the set of all attributes in the dataset; namely, the set of possible patterns that include all the attributes and every value for each attribute that appears in the data.
The experiments were executed on macOS Catalina, 64-bit, with
16GB of RAM and Intel Quad-Core i7 3.1 GHz processor.  All algorithms were implemented in Python 3.

\setlength{\belowcaptionskip}{-8pt}
\begin{figure*}
	\centering
	\includegraphics[width =\linewidth]{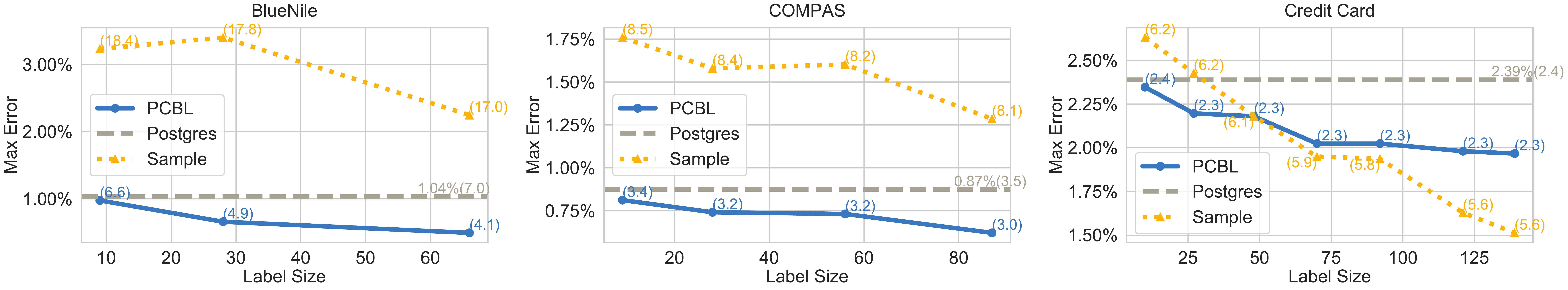}
	\caption{Absolute max error as a function of label size (mean values are shown in parenthesis)} \label{fig:label_acc}
	\vspace{-0.2cm}
\end{figure*}

\setlength{\belowcaptionskip}{-8pt}
\begin{figure*}
	\centering

		\includegraphics[width=\linewidth]{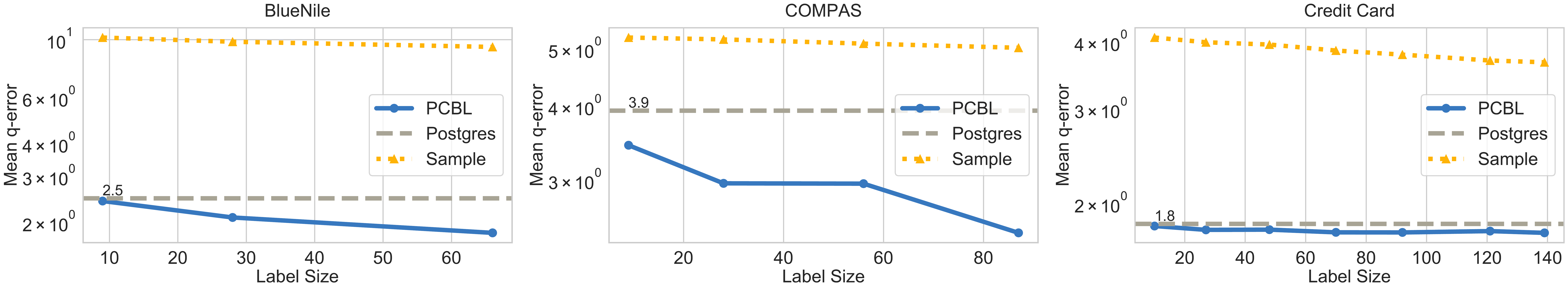}

	\caption{Mean q-error as a function of label size} \label{fig:q-err}
% 	\vspace{-0.2cm}
\end{figure*}

\subsection{Label accuracy} \label{subsec:acc}
We assessed the quality of the generated labels in estimating the data pattern's count by examining the error induced by the labels of varying size with respect to the set of patterns appearing in the database. We varied the label's size bound from 10 to 100 to generate labels with different size.
% We varied the label's size bound\footnote{From uniformity consideration the bound is over $|PC|$. The resulting label size is $|PC| + |VC|$, and as we explain in Section \ref{sec:model}, $|VC|$ is invariable and determined by dataset's structure.} from 10 to 100 to generate labels with different size.

\paragraph*{Compared Baselines} We
have measured the accuracy of our proposed pattern count based label (PCBL, blue line in the graphs) to two baseline approaches.
\begin{compactdesc}
    \item[PostgreSQL] The PostgreSQL row estimation relies on 1D histograms. It stores the statistical data about the database in {\small\texttt{pg\_statistic}} and random sampling while producing statistics. 
    \item[Sampling] Uniform random sample with growing size. The size of a sample that corresponds to the bound $x$ is $x + |VC|$\footnote{Recall that the bound $B_s$ is over the pattern count set size $|PC|$, see Section \ref{sec:model}.}. Given a sample $S$ of size $|S|$ for a dataset $D$, and a pattern $p$, we use $c_S(p)\cdot\frac{|D|}{|S|}$ to estimate the count of $p$ in $D$, where $c_S(p)$ is the count of $p$ in $S$. %The sample based estimation is marked in yellow.
\end{compactdesc}

\paragraph*{Error Measures}
We compared the quality of the estimation method with different error measures. 
\begin{compactdesc}
\item[Absolute error] Error was measured as the absolute value of difference in count between the actual and estimated count for each pattern. Recall that the absolute maximum error is  our estimation error measure (as we defined it in Section \ref{sec:model}). 
    % \item[Absolute error] Error was measured as the absolute value of difference in count between the the actual and estimated count for each pattern; we then report the {\em maximum} error over all patterns in the database. Recall that the maximum error is  our measure for error of a label (as we defined it in Section \ref{sec:model}). 
    \item[Q-error]  The factor by which an estimate differs from the actual count (see definition in Section \ref{sec:model}).
    This error measure is a standard accuracy metric in query estimation, where the accuracy is reported as mean q-error. To avoid division by zero, we set $est(p) = 1$ whenever the actual estimation was $0$.
    % \item[Q-error]  The factor by which an estimate differs from the actual count.
    % $$q\text{-}error(p) = \max\Big(\frac{c_D(p)}{est(p)}, \frac{est(p)}{c_D(p)}\Big)$$
    % This error measure is a standard accuracy metric in query estimation, where the accuracy is reported as mean q-error. To avoid division by zero, we set $est(p) = 1$ whenever the actual estimation was $0$.
\end{compactdesc}

For all three datasets, we observed similar errors for the label generated by the optimal heuristic and the one generated by the naive algorithm (blue line in the graphs). In all cases {\small\texttt{pg\_statistic}} contained over 400 rows ($429$ in the BlueNile dataset, $439$ in the COMPAS dataset, and $446$ in the Credit Card). The accuracy is independent of the label size, and is marked with a gray line in the graphs. For the sample based estimation we report the average over 5 executions and the results are marked in yellow.

Figure \ref{fig:label_acc} shows the absolute max error (mean error values are shown in parenthesis) as a function of the label size.
The maximal error is presented as a fraction of the data size.
For the BlueNile dataset the maximum estimation error was $1136$ (less than $1\%$) for a label of size $9$ (generated when setting the bound to $10$).  When setting the bound to $100$ we obtained a label of  size $66$ with maximum error of $575$ (around $0.5\%$). The postgres maximal error was 1204 ($1.04\%$) and the mean was about $7$. In the sample based estimation we observed a small increase in the maximal error for a sample of $75$ (corresponds to label with $|PC|=28$, bound of 30). This is because the sample size is significantly smaller that the database size, thus $\frac{|D|}{S}$ is larger that the count of all tuples in the data, which results in over estimation for all tuples in the sample, and estimation of $0$ for the rest. In particular, if the count of a pattern is greater than $2$ (as in one of the executions in this experiment) the overestimation is even higher. The mean error of the sample based method decreased from $18.44$ for the smallest sample size ($\times3$ of the PCBL) to $17.04$ in the largest sample (over $\times4$ of the PCBL).

% For the BlueNile dataset (Figure \ref{fig:blueNile_acc}) the maximum estimation error was $1136$ (less than $1\%$) for a label of size $9$ (generated when setting the bound to $10$).  The mean error was $6.6$ with standard deviation of $31.5$. When setting the bound to $100$ we obtained a label of  size $66$ with maximum error of $575$ (around $0.5\%$), and mean error of $4.1$ with standard deviation of $16.2$. The postgres maximal error was 1204 ($1.04\%$) and the mean was about $7$. In the sample based estimation we observed a small increase in the maximal error for a sample of 75 (corresponds to label with $|PC|=28$). This is because the sample size is significantly smaller that the database size, thus $\frac{|D|}{S}$ is larger that the count of all tuples in the data, which results in over estimation for all tuples in the sample, and estimation of $0$ for the rest. In particular, if the the count of a pattern is greater than 2 (as in one of the executions in this experiment) the overestimation is increased. The mean error of the sample based method decreased from $18.44$ for the smallest sample size ($\times3$ of the PCBL) to $17.04$ in the largest sample (over $\times4$ of the PCBL).

% The error rates for the COMPAS dataset are depicted in Figure \ref{fig:compas_acc}. 
For the COMPAS dataset, the size of the label generated when setting the bound to $10$ was $9$ and the maximum error induced by the generated labels was $494$ (about $0.8\%$). For a label of size $87$, generated with bound of $100$, the maximum error was $378$ (a little over $0.6\%$). In the postgres estimations the maximal error was $532$ ($0.87\%$) and mean error of $3.48$. The maximal error of the sample based estimation was $1070$ for the smallest sample size, and $782$ for the largest.

% The error rates for the COMPAS dataset are depicted in Figure \ref{fig:compas_acc}. The size of label generated when setting the bound to $10$ was $9$ and the maximum error induced by the generated labels was $494$ (about $0.8\%$). The mean error in this case was $3.4$ with standard deviation of $13.6$. For a label of size $87$, generated with bound of $100$, the maximum error was $378$ (a little over $0.6\%$) and mean error of $3.0$ with standard deviation of $11.3$. In the the postgres estimations the maximal error was $532$ ($0.87\%$) and mean error of $3.48$. The maximal error of the sample based estimation was $1070$ with mean of $8.45$ for the smallest sample size, and $782$ and $8.07$ for the largest. 

% Figure \ref{fig:credit_acc} shows the results for the Credit Card dataset.
The label obtained with bound of $10$ contained $10$ pattern-count pairs in the Credit Card dataset. The maximum error was $704$ ($2.3\%$). For a label with $92$ pattern-count pairs (generated with the bound set to 100), we obtain  maximum error of $607$ ($2.0\%$). The maximum observed error remains $607$ when we increased the label size bound from $70$ to $100$ (generating labels of size 70 and 92 respectively). We note that the mean error decreased $2.2978$ to $2.2974$. To further demonstrate the trend, we examine the error of labels generated with bound set to $125$ and $150$, which generated labels of size $121$ and $139$ respectively. The postgres maximal error estimation was $717$ ($2.39\%$), with mean of $2.44$. The maximal error of the average sample based estimation decreased from $789$  to $453$, which is slightly better than the results of the PCBL, however the mean error was higher, $6.25$ to $5.59$ (about $\times 3$ of PCBL). 

% Figure \ref{fig:credit_acc} shows the results for the Credit Card dataset. The label obtained with bound of $10$ contained $10$ pattern-count pairs. The maximum error was $704$ ($2.3\%$) and the mean error was $2.4$ with  
% standard deviation $16.2$. For a label with $92$ pattern-count pairs (generated with the bound set to 100), we obtain  maximum error of $607$ ($2.0\%$), mean error of $2.3$ and standard deviation of $14.9$. The maximum observed error remains $607$ when we increased the 
% the label size bound from $70$ to $100$ (generating labels of size 70 and 92 respectively). We note that the mean error decreased $2.2978$ to $2.2974$. To further demonstrate the trend, we examine the error of labels generated with bound set to $125$ and $150$, which generated labels of size $121$ and $139$ respectively. The postgres maximal error estimation was $717$ ($2.39\%$), with mean of $2.44$. The maximal error of the average sample based estimation decreased from $789$  to $453$, which is slightly better than the results of the PCBL, however the mean error was higher, $6.25$ to $5.59$ (about $\times 3$ of PCBL). 

The mean $q$-$error$ is shown in Figure \ref{fig:q-err}. In all cases, PCBL outperformed the competitors, and we observed a decrease in the error as the label size grows. For the BlueNile dataset the max $q$-$error$ for the smallest label was $47$ compared to average of $2039$ using the corresponding samples. The mean was $2.4$ and $10.2$ respectively. The PCBL max $q$-$error$ dropped down to $25$ with mean of $1.8$ using the largest label. The max error using the largest sample was $1335.4$ and the mean was $9.4$.  
The postgres maximal $q$-$error$ in this case was $45$ and the mean was $2.5$. 
In the COMPAS dataset, the max $q$-$error$ was $234$ and $715$ for the PCBL and sample methods respectively for the smallest bound, with mean of $3.4$ and $5.2$.  For the largest bound the max $q$-$error$ was $101$ and $387$, and the mean was $2.4$ and $5.02$ using PCBL and the sample respectively. In the postgres estimation we obsrved a max error of $234$ and mean of $3.9$. 
Finally, for the Credit Card dataset the observed max error was $47$ in all label's sizes using the PCBL, and the mean decreased form $1.8$ to $1.7$. Using the samples, the max error was $426.8$ and $238.6$, with mean of $4.1$ and $3.7$ for the smallest and largest samples respectively. The postgres max $q$-$error$ was $47$ and the mean was $1.8$.

\subsection{Label generation time}\label{subsec:genTime}
% \paragraph*{Label generation time}
% \yuval{report the average time of each part in the computation (cands search and minimal selection)}
The next set of experiments aims at studying the scalability of the algorithms for label generation.
We compared the performance of our proposed optimized heuristic algorithm (dark blue) to a baseline naive algorithm described in Section~\ref{sec:algo} (light blue). 
\yuval{new} The reported times for the optimized heuristic are the total generation time, including both the candidates search time and finding the best label in the candidate set. 
Since the number of patterns is large (the same size of the database), the latter may be costly. However, as we use maximal error, we were able to optimize it as follows. We sort the patterns by count in a decreasing order. Then, when traversing the patterns, we compute the error for each one, while tracking the maximal error observed. Once we reach a pattern with lower count than the observed maximal error we terminate. On average, finding the optimal label out of the candidates set was $62.6\%$ of the total running time in the BlueNile dataset, $18\%$ in COMPAS and $44.4\%$ in the Credit Card dataset.

% \begin{figure*}[h]
% 	\centering
% 	\begin{subfigure}[t]{0.33\textwidth}
% 		\includegraphics[width = 0.9\linewidth]{Figures/blue_nile_runtime_plot.pdf}
% 		\caption{BlueNile dataset}
% 		\label{fig:blueNile_label_gen_rtime}
% 	\end{subfigure}%
% 	\begin{subfigure}[t]{0.33\textwidth}
% 		\includegraphics[width = 0.92\linewidth]{Figures/compas_runtime_plot.pdf}
% 		\caption{COMPAS dataset}
% 		\label{fig:compas_label_gen_rtime}
% 	\end{subfigure}%
% 	\begin{subfigure}[t]{0.33\textwidth}
% 		\includegraphics[width = 0.95\linewidth]{Figures/credit_card_runtime_plot.pdf}
% 		\caption{Credit Card Clients dataset}
% 		\label{fig:credit_card_label_gen_rtime}
% 	\end{subfigure}
% 	\caption{Label generation runtime as a function of label size bound} \label{fig:label_gen_rtime}
% \end{figure*}

\begin{figure*}[ht]
	\centering
	\includegraphics[width =\linewidth]{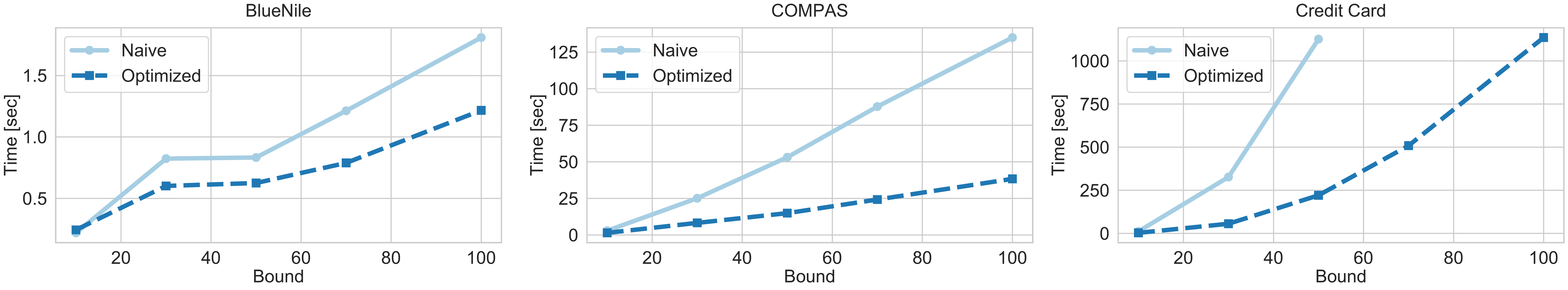}
	\caption{Label generation runtime as a function of label size bound} \label{fig:label_gen_rtime}
\end{figure*}

Figure \ref{fig:label_gen_rtime} depicts the running time as a function of the label's size bound from bound 10 and up to 100. As the bound grows, the number of possible attributes subsets that may be used to generate an optimal label increases, which affect the generation time for both algorithms. The optimized heuristic outperform the naive algorithm since the number of subsets it consider is smaller (we give the actual number of subsets in Section \ref{subsec:optimization}). In the Credit Card dataset, the naive algorithm did not terminate within $30$ minutes beyond bound of $50$. For bound of $50$ the naive algorithm running time was over $18$ minutes.  The optimal heuristics was able to compute the label for bound of $50$ with about $3.5$ minutes, and the label for the largest bound of $100$ within $18$ minutes.

\begin{figure*}
	\centering
	\includegraphics[width=\linewidth]{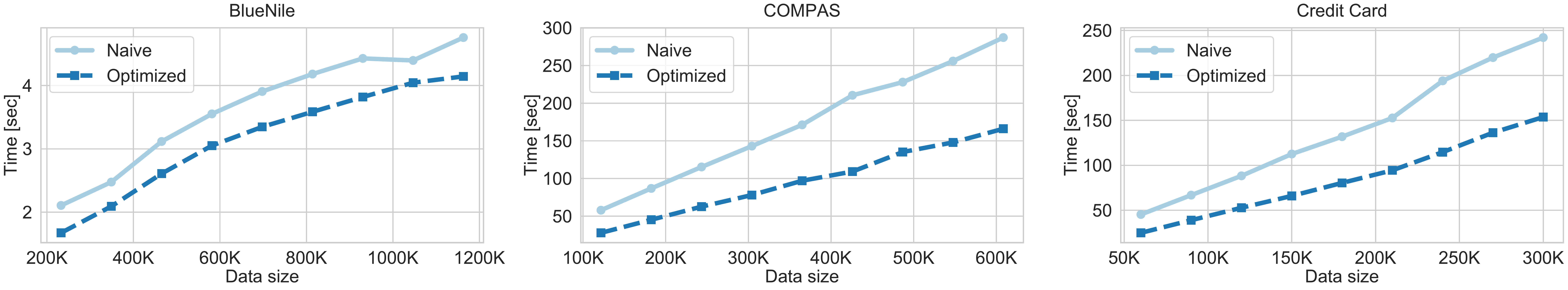}
	\caption{Label generation runtime as a function of data size} \label{fig:db_size}
\end{figure*}

%\begin{figure*}
%	\centering
%	\begin{subfigure}[t]{0.33\textwidth}
%		\includegraphics[width = 0.95\linewidth]{Figures/blue_nile_data_dup_plot.pdf}
%		\caption{BlueNile dataset}
%		\label{fig:blueNile_db_size}
%	\end{subfigure}%
%	\begin{subfigure}[t]{0.33\textwidth}
%		\includegraphics[width = 0.95\linewidth]{Figures/compas_data_dup_plot.pdf}
%		\caption{COMPAS dataset}
%		\label{fig:compas_db_size}
%	\end{subfigure}%
%	\begin{subfigure}[t]{0.33\textwidth}
%		\includegraphics[width = 0.95\linewidth]{Figures/credit_card_data_dup_plot.pdf}
%		\caption{Credit Card Clients dataset}
%		\label{fig:credit_db_size}
%	\end{subfigure}
%	\caption{Label generation runtime as a function of data size} \label{fig:db_size}
%\end{figure*}

%\begin{figure*}
%	\centering
%	\begin{subfigure}[t]{0.33\textwidth}
%		\includegraphics[width = 0.95\linewidth]{Figures/blue_nile_data_size_plot.pdf}
%		\caption{Blue Nile dataset}
%		\label{fig:blueNile_db_size}
%	\end{subfigure}%
%	\begin{subfigure}[t]{0.33\textwidth}
%		\includegraphics[width = 0.95\linewidth]{Figures/compas_data_size_plot.pdf}
%		\caption{COMPAS dataset}
%		\label{fig:compas_db_size}
%	\end{subfigure}%
%	\begin{subfigure}[t]{0.33\textwidth}
%		\includegraphics[width = 0.95\linewidth]{Figures/credit_card_data_size_plot.pdf}
%		\caption{Credit Card Clients dataset}
%		\label{fig:credit_db_size}
%	\end{subfigure}
%	\caption{Label generation runtime as a function of data size} \label{fig:db_size}
%\end{figure*}

The next experiments aim at assessing the effect of the data size (i.e., number of tuples) and the number of attributes, on the label generation time. We note that the number of attributes subsets examined by the algorithms to generate the optimal label depends (exponentially) on the number of attributes, whereas the number of tuples affects the examination time of each subset (i.e., measuring it's size and error rate). Thus, we expect to see a moderate growth in the label generation time as a function of the database size, and a steep growth in the generation time as a function of the number of attributes.

To study the effect of the data size on the algorithm's running time we gradually increased the data size by adding randomly generated tuples to the datasets. We increased the data size up to $\times 10$ the original data size. We repeated each experiments $5$ times and report the average running  time of the label generation for the bound of $50$ in Figure \ref{fig:db_size} (we observed similar trends
for other bound setting). As expected, we observed a moderate growth with respect to the data size for all three datasets.

% To study the effect of the data size on the algorithm's running time we gradually increased the data size by adding randomly generated tuples to the datasets. We increased the data size up to $\times 10$ the original data size. We repeated each experiments $5$ times and report the average running  time of the label generation for the bound of $50$ in Figure \ref{fig:db_size}. We observed similar trends
% for other bound setting.

% As expected, we observed a moderate growth with respect to the data size for all three datasets (Figure \ref{fig:db_size}).
% The results for the BlueNile dataset are shown in Figure \ref{fig:blueNile_db_size}. The label generation time increased from $2.11$ and $1.67$ seconds for the naive and optimized heuristic on a database of 232,600 tuples, up to $4.76$ and $4.14$ for a data of 1,163,000.  Similarly, in the COMPAS dataset (Figure \ref{fig:compas_db_size}), the naive algorithm's running time increased from $58$ to $287$ seconds and the optimized heuristic runtime increased from $28$ seconds for a data with 121,686 tuples to $166$ seconds for the data with 608,340 tuples. Finally, for the Credit Card dataset (Figure \ref{fig:credit_db_size}), we observed a growth from $45$ and $25$ seconds to $242$ and $153$ seconds for the naive and optimized heuristic respectively, when increasing the data from 60,000 tuples to 300,000 tuples.

Interestingly, in the Credit Card dataset, the performance of both algorithms for the dataset with 60,000 tuples (45 and 24 seconds for the naive algorithm and the optimal heuristic respectively-- first point in the rightmost graph in Figure \ref{fig:db_size}) was better than  their respective performance over the original 30,000 tuples (18 minute for the naive algorithm and 221 seconds for the optimal heuristic--3'rd point in the corresponding graph in Figure \ref{fig:label_gen_rtime}). 
%\jag{Note sure which figure this paragraph references -- surely not 7a.  I do not see any figure with a point for 60,000 tuples} \yuval{is it clear now?}.
The reason for that is that by adding new randomly generated tuples, we introduced new patterns that were missing in the original data. As a result, the number of attribute subsets examined by the algorithm, and in turn the overall running time of the algorithms, decreased.  To illustrate, the number of attribute sets examined by the naive algorithm for the original dataset was 536,130 and 9,156 for the optimized heuristic. For the date (with randomly generated tuples) of 60,000 tuples the naive algorithm examined 12,926 attribute sets, and the optimized heuristic only 785 sets.

\begin{figure*}
	\centering
		\includegraphics[width = \linewidth]{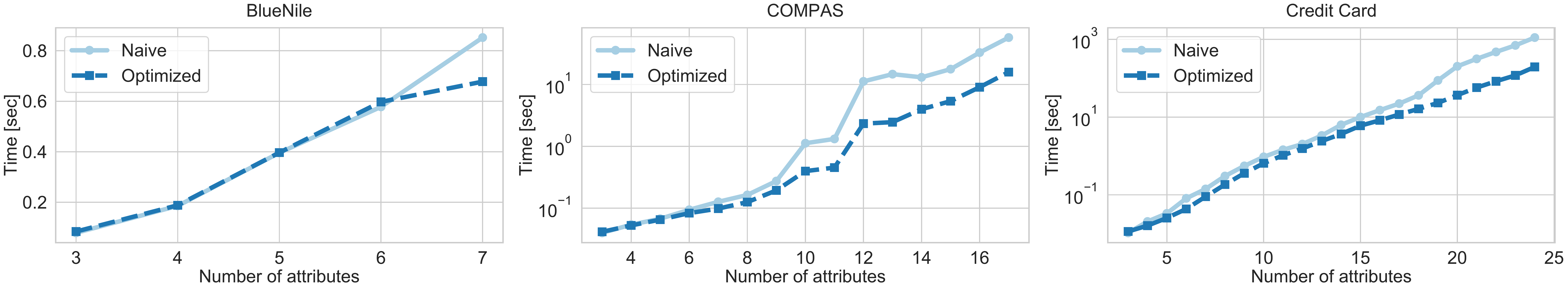}
	\caption{Label generation runtime as a function of number of attributes}
	\label{fig:attr_num}
\end{figure*}

Figure \ref{fig:attr_num} depicts the running time as a function of the number of attributes. We fixed the bound to $50$ and varied the number of attributes in the datasets from $3$ to $|\mathcal{A}|$ where $\mathcal{A}$ is the set of all attributes in the dataset. The effect on the running times was more notable in the COMPAS and the Credit Card datasets since they contain larger numbers of attributes. The results for these datasets are thus  presented in log scale.

\setlength{\belowcaptionskip}{-10pt}
\begin{figure*}
	\centering
	
		\includegraphics[width=\linewidth]{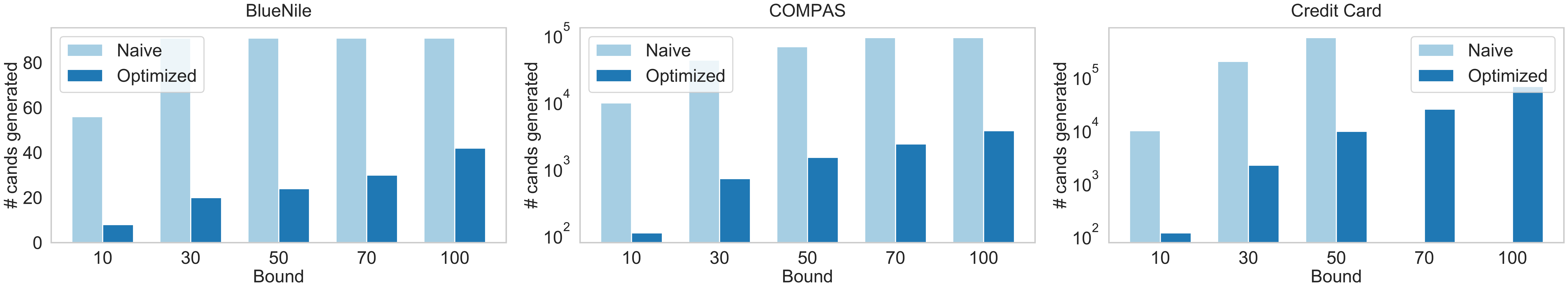}
	
	\caption{Number of labels candidates examined as a function of label size bound} \label{fig:cands_num}
\end{figure*}

\subsection{Effect of optimization}\label{subsec:optimization}
% \paragraph*{Effect of optimization} 
Recall that our heuristic optimizes the number of attribute sets examined during the search of the optimal label. To quantify the usefulness of our heuristic, we compared the number of attributes sets examined during the label generation by the optimized heuristic and the naive algorithm. We observed a gain of up to $99\%$ in the number of subsets examined as shown in Figure \ref{fig:cands_num}.  

For the BlueNile dataset we observed the lowest gain of $54\%$: $91$ subsets examined by the naive algorithm compared with $42$ by the optimized heuristic for a bound of $100$. The largest gain in this dataset was $86\%$ for a bound of $10$ ($56$ for the naive algorithm and $8$ for the optimized heuristic). For the largest bound, the naive algorithm generate $71\%$ of all possible attributes subsets, while the optimized heuristic generate only~$33\%$.

The gain in the COMPAS dataset varied from $96\%$ (89,828 compared to 3,594 for a bound of $100$), and up to $99\%$ for a bound of $10$ (9,384 by the naive algorithm compared to $106$ by the optimized heuristic). In the worst case the naive algorithm examined $69\%$ of of all possible attributes subsets, and the optimized heuristic examined only~$3\%$.

% The results for the  Credit Card dataset are shown in Figure \ref{fig:credit_cands_num}.
%Figure \ref{fig:credit_cands_num} depicts the number of attributes generated by each one of the compared algorithms for the Credit Card dataset.
For a bound of $50$, the number of subsets generated for the  Credit Card dataset by the naive algorithm was 536,130 wheres the optimal heuristic generated only 9,156 subsets, a gain of $98\%$. For a bound of $10$ and $30$ the gain was $99\%$ (9,384 compared with $112$, and  190,026 compared with 2,102 resp.). For a bound of $100$ the heuristic algorithm generated 64,312 attributes subsets, only $0.4\%$ of the total number of possible subsets (recall that the naive algorithm did not terminate within $30$ minutes beyond bound of $50$).

\subsection{Sub-labels accuracy}\label{sec:sub_lable_acc}
% \paragraph*{Sub-labels accuracy} 
The goal of our last experiment was to validate the assumption from Section \ref{sec:label_char} indeed takes place in practice. 
Namely, that the error entails from a label generated using a subset of attributes $S$ is at most the error entails by the label generated using any subset of $S$. To this end, we used the subset of attributes $S$ used to generate optimal label  (for a bound of $100$) for each dataset, and examine the error incur by the labels generated with each possible subset of $S$. 

The dark bars in Figure \ref{fig:sub_labels} depict the maximum error for the optimal label for each dataset (orange for BlueNile, green for COMPAS and purple for the Credit Card). The light bars shows the maximum error of the labels generated from the attributes sets obtained by removing a single attribute from the set used to generate the optimal label. 

For the BlueNile dataset, the optimal label was generated using the attributes {\small\texttt{cut}}, {\small\texttt{shape}} and {\small\texttt{symmetry}}. The maximum error for the label generated using this set of attributes was $0.49\%$ (the dark orange bar). The light orange bars shows the maximum error rate observed for the labels generated using the sets {\small\{\texttt{cut}, \texttt{shape}\}}, {\small\{\texttt{cut}, \texttt{symmetry}\}}, {\small\{\texttt{shape}, \texttt{symmetry}\}}. The error in all cases was higher than the error of the optimal label (from $0.8\%$ and up to $0.91\%$).

We observed similar results for the COMPAS dataset. The dark green bar shows the maximum error of the optimal label ($0.62\%$). In this case the optimal label was generated using a set of six attributes: {\small\texttt{RecSupervisionLevel}}, {\small\texttt{RecSupervisionLevelText}}, {\small\texttt{DisplayText}}, {\small\texttt{Scale\_ID}}, {\small\texttt{DecileScore}} and {\small\texttt{ScoreText}}. For
each label generated from a set obtained by removing a single attribute from the set used to generate the optimal label we obtained a label with an higher error rate (shown in light green bars) from $0.72\%$ and up to~$0.81\%$.

Finally, the optimal label for the Credit Card dataset was generated using the attribute set containing the attributes {\small\texttt{education}}, {\small\texttt{marriage}}, {\small\texttt{age}} and {\small\texttt{PAY\_AMT1}}, which describes the amount paid in September, 2005. The maximum error of the optimal label was $2.02\%$ (the dark purple bar). In three out of the four attributes subsets (light purple bars), the maximum error was higher than the optimal label (from $2.2\%$ to $2.34\%$).  The error of the label generated using only {\small\texttt{education}}, {\small\texttt{marriage}}, {\small\texttt{age}} was similar to the optimal error.

To conclude, the result of this experiment supports our claim and indicates that the assumption (that a more specific pattern count leads to lower error in the count estimate) underlying our optimized heuristic indeed holds in practice.

\begin{figure}
	\centering
	\includegraphics[width = 0.33\textwidth]{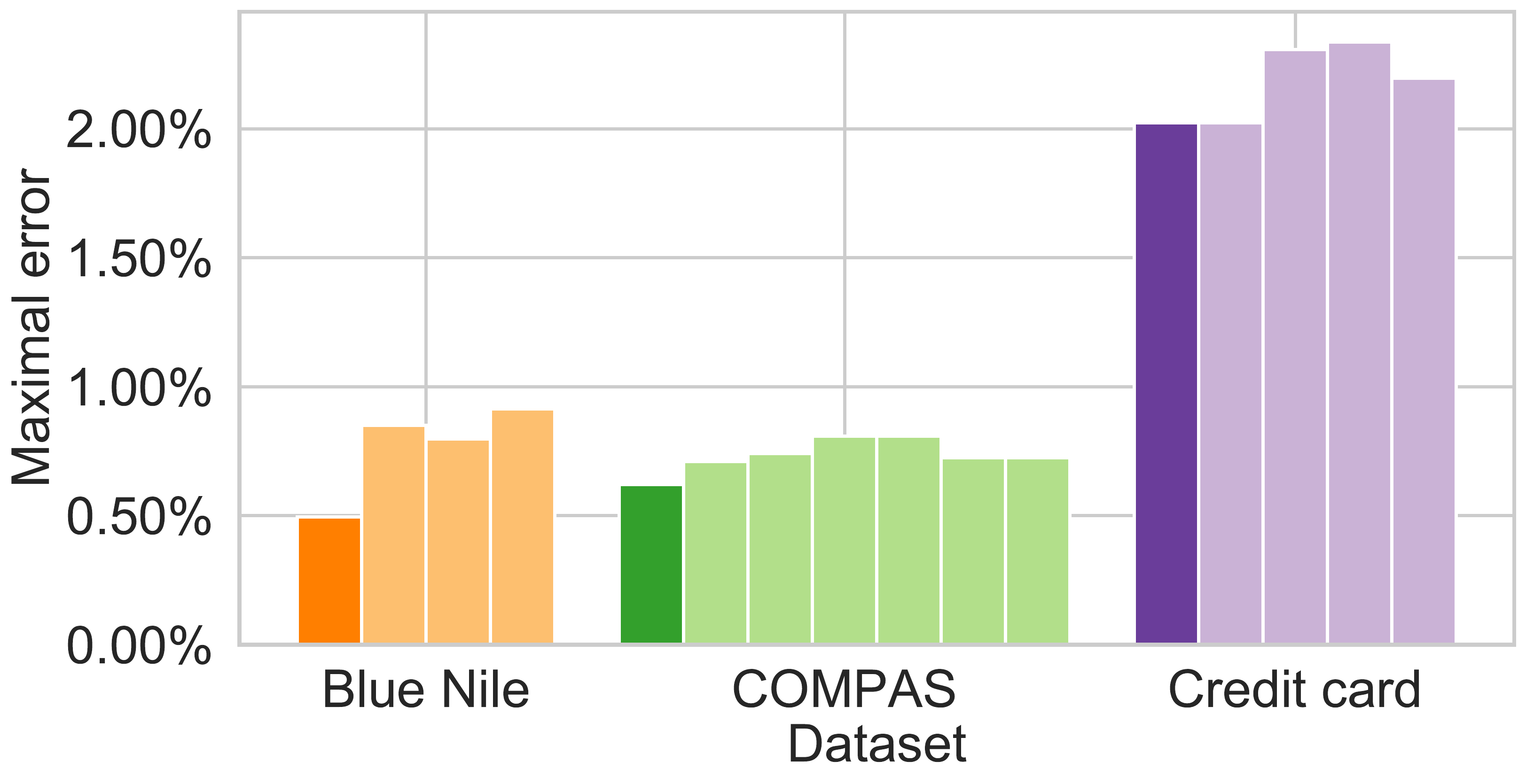}
	\caption{Optimal label vs. Sub labels error. For each data set, the dark bar indicates the performance with label bound set to 100.  The light bars represent the error of the labels generated from the attributes sets obtained by removing a single attribute from the optimal set.}
% 	\caption{Optimal label Vs. Sub labels error. For each data set, the dark bar indicates the performance with label bound set to 100.  The light bars represent the error of the labels generated from the attributes sets obtained by removing a single attribute from the set used to generate the optimal label.}
% 		\vspace{-0.2cm}
	\label{fig:sub_labels}
\end{figure}

%For the COMPAS dataset the optimal label generated with the bound set to $100$ contained the count of the patterns over the attributes \texttt{RecSupervisionLevel}, \texttt{RecSupervisionLevelText}, \texttt{Scale\_ID}, \texttt{DecileScore} and \texttt{ScoreText}

%For the Credit Card dataset, the optimal label was generated baed on the attributes education, marriage, age and \texttt{PAY\_AMT1} which describs amount paid in September, 2005.
%
%For the BlueNile and COMPAS datasets, the labels generated from subsets of the optimal set led to increasing  error rates. For the Credit card dataset, the error of the labels generated by three out of the four subsets incurred in increasing the error, and the same maximal error observed for the label generated by the last subset. 

%\subsection{Results}

% \subsection{A case study}
% We next present a case study whose goal is demonstrate the intuition underlying our proposed
% labels model and  illustrating it usefulness by getting c closer look on the resulting labels.

%% file: related.tex
%!TEX root = ./main.tex
\section{Related Work}\label{sec:related}

%We next overview multiple lines of related work.

With increasing interest in data equity in recent years, multiple lines of work have focused on labeling data and models in order to improve transparency, accountability and fairness in data science \cite{StoyanovichH19, DNL, datasheets, MithraLabel, modelCards,NLforRanking}.

% With the increase interest in data equity in recent years, multiple line of work has focused on labeling data and models with information that aims at improving transparency, accountability and fairness in data science \cite{StoyanovichH19, DNL, datasheets, MithraLabel, modelCards,NLforRanking}.

%transparency, accountability and fairness in data science in recent years multiple line of work has focused on labeling data and models with information that allow users determine fitness for use. 

Different data labeling models were studied in \cite{DNL, datasheets, MithraLabel}.
Data nutrition labels \cite{DNL} are composed of modules, called widgets. Modules are stand-alone, and each provides a different flavor of information: metadata, provenance, variables, statistics pair, probabilistic model and ground truth correlations. The models vary in the manual effort required for their generation and their technical sophistication. Overall, the labels allow users to  interrogate various aspects of the dataset.
 Our proposed label model may  be assimilated as a widget or a module in the above models. An important feature of our model is the ability to automatically generate the labels. 

Other works focused on model labeling \cite{modelCards, NLforRanking}. 
The model cards defined in \cite{modelCards} is a framework that encourage transparent model reporting. The authors proposed  a standard way of reporting information regarding machine learning models, taking into consideration the context in which they are applied, and ethical aspects. The model cards include information about the model such as how it was built, what assumptions were made and its possible effect on different ``protected" groups. %The proposed information that should be in the cards in not complete, and adjustments may be needed for different models.
The work of \cite{NLforRanking} has focused on generating nutritional label for rankings. The ranking facts of \cite{NLforRanking} is a set of widgets that present different aspects of the ranking algorithm while addressing key principles of transparency and interpretability, such as attributes with significant impact on the outcome, the stability of the ranking, its fairness with respect to different fairness measures and diversity.

While the idea of a nutritional label has been very nicely argued for in work such as that cited above, the actual content of the label is either manually generated, or at most has an aspiration towards automated generation beyond the simplest properties.  Our work establishes the first critical widget that provides substantive information about a data set and is constructed in a completely automated manner. 

Data profiling is used not just for nutrition labels, but also for many other purposes.
Most notably, database systems have used such information for decades to assist in query optimization through query result size estimation. Histograms are commonly used on individual attributes~\cite{JagadishKMPSS98}. However, histograms on more than one attribute at a time are uncommon. 
The problem of query size estimation based on multi-dimensional histograms was studied in \cite{LeeKC99, DeshpandeGR01,ThaperGIK02,BrunoCG01}. These are not restricted to categorical datasets, but they work well only for low to medium data dimensionalities (typically 2-5 and at most 12, see~\cite{DeshpandeGR01}).
Our work in this paper can also be of value in building better multi-dimensional histograms.

\yuval{new}
There is a wealth line of work on selectivity estimation \cite{CormodeGHJ12, MullerMK18, DuttWNKNC19, YangLKWDCAHKS19} using various methods from sampling and synopses \cite{MullerMK18} to machine learning and deep learning that were suggested in recent works (e.g., \cite{DuttWNKNC19, YangLKWDCAHKS19}). Sampling methods are typically simple to implement but they are sensitive to skew and have insufficient performance for high selectivity queries, both crucial for the intended labels applications. As indicated by our experiments, using small samples (of same size as our labels) results in poor estimations. While machine learning based methods often do remarkably well, the resulting models are very complex and have a higher memory consumption than our proposed labels. Moreover, our proposed approach is derived from a user perspective, and designed to allow for human visualization and interpretation. The typically complexity of ML models makes them ill-suited for such purpose.

Our proposed label model may be reminiscent of the minimum description length (MDL) principle  \cite{Rissanen, MDL}, an important concept in information theory and computational learning theory. The MDL principle addresses the problem of  choosing the model that gives the shortest description of data. At a high level, the idea behind MDL is that the model that best captures or fits the important features of the data
is the one that is able to compress the data most. The basic idea is then to use two parts to describe the data: the hypothesis (or model) and an encoding of the data using that model. In a way, our problem may be considered as an MDL problem, where the label is the model and the set of errors with respect to each pattern as the additional information needed for the description of the data given the model.
An inherent difference between our work and the MDL principle, is that in our proposed model we aim at minimizing the error within a given bound limit, whereas in the  MDL principle the goal is to minimize the total description length (and not only the error).

The problem of reconstructing finer-scale data from multiple coarse views, aggregated over different (subsets of) dimensions was presented in \cite{PREMA}. The goal of the disaggregation task is to estimate a particular
series in a higher resolution, given observations in lower
resolution.
While our estimation technique relies on a single aggregated instance, that best estimates the original data, the work of \cite{PREMA} uses multiple aggregate data to reconstruct the original data.
 
There is a wealth of work on lossy data compression \cite{ZhangW96, YangK96}. Various techniques were proposed for different application such as image compassion \cite{AnsariMC98} and text compression \cite{WittenBMST94, TextCompressionBook}. While our proposed model of data labels may be considered as a new lossy  data compression method, our intended usage of the labels is different and as a result we do not consider the decoding process of the entire compressed data in bulk.

%% file: conc.tex
\section{Conclusion}
\label{sec:conc}

We have developed a ``label" for a data set that can be used to determine the count for every pattern of attribute value combinations in the data set. Since these counts are typically central to determining fitness for use, and thus avoid generating biased models and data-driven algorithms, our work is in line with the many recent proposals for a data set label that allows users to determine fitness for use and build trust.  
Our labels can be fully automatically generated. Our label model is built upon an estimation function, that allows to estimate the count of every pattern, using partial count information in the label. We present an optimized heuristic for optimal label generation, and experimentally show the quality of our label and usefulness of our heuristic compared with a naive algorithm.

Given the label of a found dataset, in case we observe an undesirable property of the data, such as insufficient diversity or groups with inadequate representation, the next step for a data scientist would be to determine whether the data can be adjusted in a way that will fit their chosen tasks. For instance, the work of \cite{AsudehJJ19} proposed an approach for coverage enhancement for patterns with inadequate representation through data acquisition. 
%However, collecting new data may be costly, and even impossible for the user. An intriguing direction to explore in future research is the development of  data adjust techniques, i.e., manipulating the given data so that the  unwanted properties are absent while keeping the data authenticity.

%% file: reduction.tex
%!TEX root = ./main.tex

% \section{Lower Bound}

\subsection{Proof of Theorem \ref{prop:np-hard}}
\label{sec:reduction}

% In this section we prove that the decision problem corresponding to the optimal label problem is NP-hard.
We prove Theorem \ref{prop:np-hard} via a reduction from the \emph{vertex cover} problem, a decision  problem which we now define.
For notational simplicity, and without loss of generality, we omit certain easy cases from the vertex cover problem. Namely, we require that the input graph contains at least two nodes and one edge and 
forbid self loops.

\begin{definition}[Vertex cover]
	Let $G=(V,E)$ be an undirected graph.
	A set $V' \subseteq V$  is a \emph{vertex cover} of $G$ if for every edge $\{x,y\}\in E$ 
	either $x \in V'$ or $y \in V'$.
\end{definition}

\begin{theorem}[\cite{Garey:1990}]\label{the:VC}
	Given an undirected graph $G=(V,E)$,
	where $V=\{v_1,\ldots,v_n\}$ for some $1<n$, 
	$E \neq \emptyset$, and  
	for any edge $\{v_i,v_j\}\in E$ it holds that $i \neq j$. 
	Determining if $G$ has a vertex cover $V'$ such that $|V'|\leq k$, when
	$k\in \{2,...,|V|-1\}$, is NP-hard. 
\end{theorem}

\paragraph*{Reduction} Given an input for the vertex cover problem, a graph $\mathcal{G} = (V, E)$ and $k$, we generated the following input to the optimal label problem: 
\begin{compactitem}
	\item A database $D$ with $|V| + 1$ attributes $\mathcal{A}$: $A_i$ for each $v_i\in V$ and an attribute $A_E$.
	\item For each attribute $A_i$ there are two possible values $x_1$ and $x_2$.
	\item The domain of the attribute $A_E$ contains $|E|$ possible values $x_i$ for each $e_i\in E$.
	%\jag{I do not think you need $e_i'$ here.  Any way, it is undefined.} \yuval{I think you are right, I tried simplify the proof}
	\item For each edge $e_r = \{v_i,v_j\}\in E$ there are $4\cdot|E|$ tuples in the database $D$, each containing only values for the attributes $A_E$, $A_i$ and $A_j$ (and the rest are missing values): for each $p,q\in \{1,2\}$, $|E|$ tuples such that $A_i = x_p$, $A_j = x_q$ and $A_E=x_r$
	\item For each $v_i,v_j\in V$ such that $i\neq j$:
	\begin{compactenum}
	    \item if $\{v_i,v_j\}\not\in E$ there are $4\cdot|E|$ tuples in the database $D$: for each $p,q\in \{1,2\}$, $|E|$ tuples such that $A_i = x_p$, and $A_j = x_q$.
	    \item if $\{v_i,v_j\}\in E$ there are $4\cdot|E|^2$ tuples in the database $D$: for each $p\in \{1,2\}$, $2\cdot|E|^2$ tuples such that $A_i = x_p$, and $A_j = x_p$.
	\end{compactenum}
	\item $B_s = 2\cdot|E| + 4\cdot \sum_{i=1}^{k-1}i$
	\item The set $\mathcal{P}$ consist of $|E|$ patterns: a pattern $p = \{A_i = x_1, A_j = x_1, A_E = x_r\}$ for each edge $e_r=\{v_i,v_j\}\in~E$.
	\item $B_e = 0$
\end{compactitem} 

\begin{example}
	Given the graph shown in Figure \ref{fig:reduction_graph} the reduction's output contains the database depicted in Figure~\ref{fig:reduction_db}. The tuples in the top left-hand side of the Figure correspond to the edge $e_1=\{v_1,v_2\}\in E$ and the tuples in the top right-hand side to the edge $e_2=\{v_2,v_3\}\in E$. The cont attribute represent the number of occurrences of each tuple in the database. The tuples in the bottom part are added because there is no edge between $v_1$ and  $v_3$. 
	The set $\mathcal{P}$ in this example contains the patterns $\{A_E = x_1, A_1 = x_1, A_2 =x_1 \}$ and $\{A_E = x_2, A_2 = x_1, A_3 = x_1\}$.
\end{example}

%\begin{example}
%	Given the graph shown in Figure \ref{fig:reduction_graph} and $k=2$ the reduction's output contains the database depicted in Figure \ref{fig:reduction_db}. The tuples in the left-hand side of the Figure correspond to the edge $\{y,z\}\in E$ and the tuples in the right-hand side to the edge $\{x,z\}\in E$. The label size bound of the reduction output instance is $B_s=2$ and the error bound is $B_e=0$. The set $\mathcal{P}$ contains the patterns $\{(x_1,y_1,\bot), (x_1,y_2,\bot), (x_2,y_1,\bot), (x_2,y_2,\bot), (x_1,\bot,z_1), \\(x_1,\bot,z_2), (x_2,\bot,z_1), (x_2,\bot,z_2)\}$.
%\end{example}
\begin{figure}
	\centering
	\includegraphics[scale=1.2]{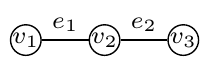}	
	\caption{Reduction input graph example}
	\label{fig:reduction_graph}
\end{figure}

\begin{figure}
	\centering
	\scriptsize
	\begin{tabular}{cc}
		\begin{tabular}{ccccc}
			\hline
			$A_E$& $A_1$     & $A_2$     & $A_3$ & count           \\  \hline \hline
			$x_1$&$x_1$ & $x_1$ &  & $2$    \\
			$x_1$&$x_1$ & $x_2$ &  & $2$   \\
			$x_1$&$x_2$ & $x_1$ &  & $2$  \\
			$x_1$&$x_2$ & $x_2$ &  & $2$ \\
			&$x_1$ & $x_1$ &  & $8$  \\
			&$x_2$ & $x_2$ &   & $8$  \\
			\hline
		\end{tabular}
		&
	\begin{tabular}{ccccc}
		\hline
		$A_E$& $A_1$     & $A_2$     & $A_3$ & count           \\  \hline \hline
		$x_2$&&$x_1$ & $x_1$ &   $2$    \\
		$x_2$&&$x_1$ & $x_2$ &   $2$   \\
		$x_2$&&$x_2$ & $x_1$ &   $2$   \\
		$x_2$&&$x_2$ & $x_2$ &   $2$  \\
		&&$x_1$ & $x_1$ &   $8$  \\
		&&$x_2$ & $x_2$ &   $8$  \\
		\hline
	\end{tabular}
	\end{tabular}
\bigskip

	\begin{tabular}{ccccc}
	\hline
	$A_E$& $A_1$     & $A_2$     & $A_3$ & count           \\  \hline \hline
	&$x_1$ && $x_1$ &   $2$    \\
	&$x_1$ && $x_2$ &   $2$   \\
	&$x_2$ && $x_1$ &   $2$    \\
	&$x_2$ && $x_2$ &   $2$   \\
	\hline
\end{tabular}
	\caption{Reduction example's output database}
	\label{fig:reduction_db}
\end{figure}

%\begin{figure}
%	\centering
%	\scriptsize
%	\begin{tabular}{cc}
%		\begin{tabular}{ccc}
%			\hline
%			X     & Y     & Z           \\  \hline \hline
%			$x_1$ & $y_1$ &      \\
%			$x_1$ & $y_2$ &     \\
%			$x_1$ & $y_3$ &     \\
%			$x_2$ & $y_1$ &    \\
%			$x_2$ & $y_2$ &    \\
%			$x_2$ & $y_3$ &      \\
%			$x_3$ & $y_1$ &      \\
%			$x_3$ & $y_2$ &    \\
%			$x_3$ & $y_3$ &    \\
%			$x_3$ & $y_3$ &     \\
%			$x_1$ &  &      \\
%			$x_2$ & &    \\
%			& $y_1$ &   \\
%			& $y_2$ &   \\
%			\hline
%		\end{tabular}
%		&
%		\begin{tabular}{ccc}
%			\hline
%			X     & Y     & Z           \\  \hline \hline
%			$x_1$ &  & $z_1$    \\
%			$x_1$ &  & $z_2$    \\
%			$x_1$ &  &$z_3$    \\
%			$x_2$ &  & $z_1$  \\
%			$x_2$ &  & $z_2$     \\
%			$x_2$ &  & $z_3$    \\
%			$x_3$ &  &$z_1$    \\
%			$x_3$ &  & $z_2$  \\
%			$x_3$ &  & $z_3$    \\
%			$x_3$ &  & $z_3$    \\
%			$x_1$ &  &      \\
%			$x_2$ &  &    \\
%			&  & $z_1$  \\
%			&  & $z_2$       \\\hline
%			
%		\end{tabular}
%	\end{tabular}
%	\caption{Reduction example's output database}
%	\label{fig:reduction_db}
%\end{figure}
To prove the reduction correctness we show that there is vertex cover of size $k$ in a graph if and only if there is a label of size at most $2\cdot|E| + 4\cdot \sum_{i=1}^{k-1}i$ with error $0$.
\begin{proposition}\label{prop:reduction}
	Given an input for the vertex cover problem, a graph $\mathcal{G} = (V, E)$ and $k$, let $D$ be the database resulting from the reduction and $\mathcal{P}$ the set of patterns.  There exists  $S\subseteq\mathcal{A}$  an attributes subset with $|L_S(D)|\leq  2\cdot|E| + 4\cdot \sum_{i=1}^{k-1}i$ such that  $Err(L_S(D), \mathcal{P}) = 0 \iff$  
	there exists  a vertex cover of size $k$ in $\mathcal{G}$.
\end{proposition}
%Let $S\subseteq\mathcal{A}$ be an attributes subset. We show that\\ $Err(L_S(D), \mathcal{P}) = 0 \iff$ 
%the set of nodes corresponding to the attributes in $S$ is a vertex cover.

To prove Proposition \ref{prop:reduction} we show that (i) the error of a label generated using a subset of attributes $S$ is $0$ is and only if the corresponding set of nodes in the graph are vertex cover using Lemma \ref{lem:err} and (ii) prove the size bounds using Lemma \ref{lem:size}.

\begin{lemma} \label{lem:err} Let $e_r = \{v_i, v_j\}\in E$,  $p=\{A_E = x_r, A_i = x_1, A_j=x_1\}$ be a pattern in $\mathcal{P}$ and   $S\subseteq\mathcal{A}$ be an attributes subset. $Err(L_S(D)) = 0 \iff A_E\in S$ and at least one of $A_i$ or $A_j$ in $S$.
\end{lemma}

\begin{proof}
	Let $e_r = \{v_i, v_j\}\in E$, $p=\{A_E = x_r, A_i = x_1, A_j=x_1\}$  a pattern in $\mathcal{P}$ and   $S\subseteq\mathcal{A}$ be an attributes subset. Note that:
	\begin{compactitem}
		\item $c_D(p) = |E|$
		\item  For each attribute $A_p$:
		$$\frac{c_D(\{A_p = x_1\})}{c_D({A_p = x_1})+c_D({A_p = x_2})} = \frac{1}{2}$$
		\item For each $x_r\in Dom(A_E)$:
		$$\frac{c_D(\{A_E = x_r\})}{\sum_{j = 1}^{|E|} c_D(\{A_E = x_j\})} = \frac{4\cdot|E|}{\sum_{j = 1}^{|E|} 4\cdot|E|} = \frac{1}{|E|}$$
		\item $|D| = 4\cdot|E|^2+4\cdot|E|^3 + 4\cdot|E|\cdot(|V|^2-|E|)$
	\end{compactitem}
	
We consider all possible cases as follows.
	\begin{compactitem}
		\item Without loss of generality assume that $A_E\in S$ and $A_i\in S$ then $$c_D(p|_{S}) = c_D(\{A_E = x_r, A_i = x_1\}) = 2\cdot|E|$$ and thus $$Est(p, L_S(D)) =2\cdot|E|\cdot\frac{1}{2} = |E|$$ Namely, the error in this case is $0$.
		\item If $A_i\in S$ and $A_j\in S$ but $A_E\not\in S$ we get $c_D(p|_{S}) = c_D(\{A_i= x_1, A_j = x_1\}) = |E|+2\cdot|E|^2$ and thus $$Est(p, L_S(D)) = |E|+2\cdot|E|^2\cdot\frac{1}{|E|}  = 2|E|+1$$ Namely, the error in this case is $|E|+1>0$.
		\item Otherwise we get  
		\begin{multline*}
		    Est(p, L_S(D)) = |D|\cdot \frac{1}{|E|}\cdot\frac{1}{2}\cdot\frac{1}{2} =\\ (4\cdot|E|^2+4\cdot|E|^3 + 4\cdot|E|\cdot(|V|^2-|E|)\cdot\frac{1}{4\cdot|E|}  =\\ |E|^2 + |V|^2 > |E|
		\end{multline*}
		Thus, the error in this case is greater than $0$.
% 		\item Otherwise we get  $Est(p, L_S(D)) = |D|\cdot \frac{1}{|E|}\cdot\frac{1}{2}\cdot\frac{1}{2} = (4\cdot|E|+4\cdot|E|^2 + 4\cdot|E|\cdot(|V|^2-|E|)\cdot\frac{1}{|E|}  = 1+|E| + |V|^2-|E| = 1+|V|^2 > |E|$. Thus, the error in this case greater than $0$.
%		\item If $A_E\in S$ but $A_i\not\in S$ and $A_j\not\in S$ we get $c_D(p|_{S}) = c_D(\{A_E= e_r\}) = 10\cdot9\cdot|E|$ and thus $Est(p, L_S(D)) = 10\cdot 9\cdot|E|\cdot\frac{1}{3}\cdot\frac{1}{3}  = 10\cdot|E|$. Thus, the error in this case greater than $0$.
%		\item If, without loss of generality. only $A_i\in S$ we get $c_D(p|_{S}) = c_D(\{A_i= x_1\}) = 4\cdot9\cdot|E|$ and thus $Est(p, L_S(D)) = 4\cdot 9\cdot|E|\cdot\frac{10}{12}\cdot\frac{1}{3}  = 10|E|$. Thus, the error in this case greater than $0$.
%		\item Finally, if  $A_i, A_j, A_E\not\in S$ we get $c_D(p|_{S}) = c_D(\{A_i= x_1\}) = 4\cdot9\cdot|E|$ and thus $Est(p, L_S(D)) = 4\cdot 9\cdot|E|\cdot\frac{10}{12}\cdot\frac{1}{3}  = 10|E|$. Thus, the error in this case greater than $0$.
	\end{compactitem}
\end{proof}

\begin{corollary}\label{cor:err}
	Let $\mathcal{P}$ be the patterns set generated by the reduction and 	let  $S\subseteq\mathcal{A}$ be an attributes subset
	\begin{multline*}
		Err(L_S(D), \mathcal{P}) = 0 \iff A_E\in S \text{ and }\\ \text{ at least one of } A_i \text{ or } A_j \text{ in } S \text{ for each }\{v_i,v_j\}\in E
	\end{multline*}
\end{corollary}

\begin{lemma}\label{lem:size}
	Let $S\subseteq\mathcal{A}$ be an attributes subset of size $|S| = k+1$ for $k\geq 1$ such that $A_E\in S$ then $|L_S(D)| = 2\cdot|E'| + 4\cdot \sum_{i=1}^{k-1}i$, where $E' = \{e_r = \{v_i,v_j\}\mid A_i\in S \text{ or } A_j\in S \text{ (or both)}\}$.
\end{lemma}

\begin{proof}
	The proof by induction on $k$. 
	\begin{compactdesc}
		\item [Base] If $k = 1$ then $S = \{A_E, A_i\}$. Let $E' = \{ \{v_i,v_j\}\in E\mid \forall v_j \in V\}$, by the reduction construction $D$ (and thus also $L_S(D)$) contains the patterns $\{A_E = x_r, A_i = x_1\}$ and $\{A_E = x_r, A_i = x_2\}$ for each $e_r\in E'$,  thus $|L_S(D)| = 2\cdot|E'|$ and the proposition holds.
		\item [Inductive step] Assuming the proposition holds for $k > 1$. Let $S = \{A_E, A_{i_1},\ldots, A_{i_{(k+1)}}\}$, $S' = \{A_E, A_{i_1},\ldots, A_{i_{k}}\}$, $E' = \{e_r = \{v_i,v_j\}\mid A_{i}\in S \text{ or } A_j\in S \text{ (or both)}\}$ and $E'' = \{e_r = \{v_i,v_j\}\mid A_{i}\in S' \text{ or } A_j\in S' \text{ (or both)}\}$. From the induction hypothesis $|L_{S'}(D)| = 2\cdot|E''| + 4\cdot \sum_{i=1}^{k-1}i$.
		Adding the attribute $A_{i_{(k+1)}}$ to the label increase the number of patterns by $4$ for each $A_{i_q}\in S'$:
		\begin{compactitem}
			\item If $e_r = \{v_{i_q}, v_{i_{(k+1)}}\}\in E$ then instead of $2$ patterns $\{A_E = x_r, A_{i_q} = x_p\}$ for each $p\in\{1,2\}$ we have $6$ patterns in $L_S(D)$: $\{A_E = x_r, A_{i_q} = x_p, A_{i_{(k+1)}} = x_m\}$, for each $p, m\in\{1,2\}$ (4 patterns), $\{A_{i_q} = x_1 , A_{i_{(k+1)}} =x_1\}$ and $\{A_i = x_2,\\ A_{i_{(k+1)}} =x_2\}$.
			\item If $\{v_{i_q}, v_{i_{(k+1)}}\}\not\in E$, the patterns $\{A_{i_q} = x_p,\\ A_{i_{(k+1)}} = x_p\}$ for each $p\in\{1,2\}$ are in $L_S(D)$ (and not in $L_S'(D)$).
		\end{compactitem}
		
	In addition, for every $e_r\in E'\setminus E''$, $L_S(D)$ contains $2$ additional patterns: $\{A_E = x_r, A_{i_{(k+1)}} = x_1\}$ and $\{A_E = x_r, A_{i_{(k+1)}} = x_2\}$. Namely%, the total number of patterns in $L_S(D)$ is
	\begin{multline*}
		|L_S(D)| = |L_{S'}(D)| + 2\cdot |E'\setminus E''| + 4\cdot k = \\ 2\cdot|E''| + 4\cdot \sum_{i=1}^{k-1}i  + 2\cdot |E'\setminus E''| + 4\cdot k= \\2\cdot |E'| + 4\cdot \sum_{i=1}^{k}i
	\end{multline*}
	\end{compactdesc}
\end{proof}

%\begin{lemma}
%	If $\{v_{i_1}, \ldots, v_{i_k}\}$ is a vertex cover in $G$ then $Err(L_S(D), \mathcal{P})=0$ for $S = \{A_E, A_{i_1},\ldots, A_{i_k}\}$.
%\end{lemma}
%
%\begin{proof}
%	
%\end{proof}

\begin{proof}{(Proposition \ref{prop:reduction})}
	Given an input for the vertex cover problem, a graph $\mathcal{G} = (V, E)$ and $k$, let $D$ be the database resulting from the reduction and $\mathcal{P}$ the set of patterns.
	%The proof follows directly from
	Assume that there exists a vertex cover of size $k$ in $\mathcal{G}$, $V'=\{v_{i_1},\ldots, v_{i_1}\}\subseteq V$. Let $S = \{A_E, A_{i_1}, \ldots, A_{i_k}\}\subseteq\mathcal{A}$ a subset of attributes. Since $V'$ is a set cover, for every edge $\{x,y\}\in E$ 
	either $x \in V'$ or $y \in V'$. Thus, from Corollary \ref{cor:err} the error of $Err(L_S(D), \mathcal{P})$ is $0$. Moreover, from Lemma \ref{lem:size}, the size of $L_S(D)$ is $2\cdot|E| + 4\cdot \sum_{i=1}^{k-1}i$.
	
	Assume that there exists a subset of attributes \\$S = \{A_E, A_{i_1},\ldots, A_{i_m}\}$ such that $|L_S(D)|\leq  2\cdot|E| + 4\cdot \sum_{i=1}^{k-1}i$ and  $Err(L_S(D), \mathcal{P}) = 0$. From  Corollary \ref{cor:err}, $A_E\in S$. Let $V' = \{ v_{i_1},\ldots, v_{i_m}\}\subseteq V$. We show that $V'$ is a vertex cover of size at most $k$. From Corollary \ref{cor:err} $V'$ is a vertex cover. Assume by contradiction that $m>k$, then from Lemma \ref{lem:size} the size of $L_S(D)$ is $2\cdot|E| + 4\cdot\sum_{i=1}^{m-1} > 2\cdot|E| + 4\cdot\sum_{i=1}^{k-1}$. Therefore $V'$ is a vertex cover and $|V'|\leq k$.
\end{proof}

The proof of Theorem \ref{prop:np-hard} follows immediately from Proposition \ref{prop:reduction} and Theorem \ref{the:VC}.

\subsection{Proof of Proposition \ref{prop:subset_acc}}

\begin{proof}
	Let $D$ be a database with attributes $\mathcal{A}$,  $S_1\subseteq S_2\subseteq\mathcal{A}$ two attribute sets, $l_i=L_{S_i}(D)$ the labels of $D$ using $S_i$ for $i=1,2$ respectively	
	and $p$ a pattern 
	such that $Attr(p)\not\subseteq S_2$. 
	Denoting $p' = p|_{Attr(p)\cap S_2}$, the pattern resulting when restricting $p$ to include only the attributes appearing in $S_2$, the estimate of $p'$ using $l_1$ is 
	\begin{multline*}
	Est(p',l_1) = \\c_D(p'|_{S_1})\cdot\prod_{A_i\in Attr(p')\setminus S_1}\frac{c_D(\{A_i = p'.A_i\})}{\sum_{a_j\in Dom(A_i)} c_D(\{A_i = a_j\})}
	\end{multline*}
	and the estimate of $p$ using $l_2$ is 
	\begin{multline*}
	Est(p,l_2) = \\c_D(p|_{S_2})\cdot\prod_{A_i\in Attr(p)\setminus S_2}\frac{c_D(\{A_i = p.A_i\})}{\sum_{a_j\in Dom(A_i)} c_D(\{A_i = a_j\})}
	\end{multline*}
	Without loss of generality, assume that  $Est(p',l_1) < c_D(p')$, namely, the estimate of $p'$ using $l_1$ is an under estimation. Lets further assume that the estimate of $p$ using $l_2$ is an under estimation, thus $Est(p,l_2) < c_D(p)$. Note that $p'|_{S_1} = p|_{S_1}$ and $p'|_{S_2} = p|_{S_2}$, thus $c_D(p'|_{S_1}) = c_D(p|_{S_1})$ and $c_D(p'|_{S_2}) = c_D(p|_{S_2})$. Moreover, since $Attr(p')\subseteq S_2$, $p'|_{S_2} = p'$ and $c_D(p') = c_D(p'|_{S_2})$  We have that 
	\begin{multline*}
	Est(p,l_1) = \\ c_D(p|_{S_1})\prod_{A_i\in Attr(p)\setminus S_1}\frac{c_D(\{A_i = p.A_i\})}{\sum_{a_j\in Dom(A_i)} c_D(\{A_i = a_j\})} = \\ c_D(p'|_{S_1})\prod_{A_i\in S_2\setminus S_1}\frac{c_D(\{A_i = p.A_i\})}{\sum_{a_j\in Dom(A_ i)} c_D(\{A_i = a_j\})} \cdot  \\
	\prod_{A_i\in Attr(p)\setminus S_2}\frac{c_D(\{A_i = p.A_i\})}{\sum_{a_j\in Dom(A_i)} c_D(\{A_i = a_j\})} =\\
	Est(p',l_1)\cdot \prod_{A_i\in Attr(p)\setminus S_2}\frac{c_D(\{A_i = p.A_i\})}{\sum_{a_j\in Dom(A_i)} c_D(\{A_i = a_j\})} 
	<\\  c_D(p')\prod_{A_i\in Attr(p)\setminus S_2}\frac{c_D(\{A_i = p.A_i\})}{\sum_{a_j\in Dom(A_i)} c_D(\{A_i = a_j\})} =\\ c_D(p|_{S_2})\prod_{A_i\in Attr(p)\setminus S_2}\frac{c_D(\{A_i = p.A_i\})}{\sum_{a_j\in Dom(A_i)} c_D(\{A_i = a_j\})} =\\ Est(p, l_2)
	\end{multline*}
	Since $Est(p, l_2) <  c_D(p)$ and $Est(p,l_1) < Est(p, l_2)$ we get $Err(l_1,p) > Err(l_2,p)$.
\end{proof}